\documentclass[a4paper,USenglish,cleveref, autoref, thm-restate]{lipics-v2019}

\usepackage{mathtools}
\usepackage[textsize=small]{todonotes}
\usepackage{subcaption} 
\usepackage[utf8]{inputenc}

\newif\iflong
\longtrue

\newcommand{\lng}[1]{\iflong#1\fi}
\newcommand{\shrt}[1]{\iflong\else#1\fi}

\newif\ifcomment
\commenttrue

\DeclareMathOperator*{\argmax}{arg\,max}
\DeclareMathOperator*{\argmin}{arg\,min}
\DeclareMathOperator*{\sumT}{\text{sum}}
\DeclareMathOperator*{\ecc}{\text{ecc}}

\DeclareMathOperator*{\id}{\text{id}}

\title{Fast Hybrid Network Algorithms for Shortest Paths in Sparse Graphs}

\author{Michael Feldmann}{Paderborn University, Germany}{michael.feldmann@upb.de}{}{}
\author{Kristian Hinnenthal}{Paderborn University, Germany}{krijan@mail.upb.de}{}{}
\author{Christian Scheideler}{Paderborn University, Germany}{scheideler@upb.de}{}{}

\authorrunning{M. Feldmann, K. Hinnenthal, C. Scheideler}

\Copyright{Michael Feldmann, Kristian Hinnenthal, Christian Scheideler}



\iflong{\relatedversion{A full version of the paper is available at \url{https://arxiv.org/abs/2007.01191}.}} 

\supplement{}

\funding{This work is supported by the German Research Foundation (DFG) within the CRC 901 ''On-The-Fly Computing'' (project number 160364472-SFB901).}

\nolinenumbers 

\lng{\hideLIPIcs}

\EventEditors{Quentin Bramas, Rotem Oshman, and Paolo Romano}
\EventNoEds{3}
\EventLongTitle{24th International Conference on Principles of Distributed Systems (OPODIS 2020)}
\EventShortTitle{OPODIS 2020}
\EventAcronym{OPODIS}
\EventYear{2020}
\EventDate{December 14--16, 2020}
\EventLocation{Strasbourg, France  (Virtual Conference)}
\EventLogo{}
\SeriesVolume{184}
\ArticleNo{31}

\begin{document}

\maketitle

\begin{abstract}
We consider the problem of computing shortest paths in \emph{hybrid networks}, in which nodes can make use of different communication modes. For example, mobile phones may use ad-hoc connections via Bluetooth or Wi-Fi in addition to the cellular network to solve tasks more efficiently. Like in this case, the different communication modes may differ considerably in range, bandwidth, and flexibility. We build upon the model of Augustine et al. [SODA '20], which captures these differences by a \emph{local} and a \emph{global} mode. Specifically, the local edges model a fixed communication network in which $O(1)$ messages of size $O(\log n)$ can be sent over every edge in each synchronous round. The global edges form a clique, but nodes are only allowed to send and receive a total of at most $O(\log n)$ messages over global edges, which restricts the nodes to use these edges only very sparsely.

We demonstrate the power of hybrid networks by presenting algorithms to compute Single-Source Shortest Paths and the diameter very efficiently in \emph{sparse graphs}. Specifically, we present exact $O(\log n)$ time algorithms for cactus graphs (i.e., graphs in which each edge is contained in at most one cycle), and $3$-approximations for graphs that have at most $n + O(n^{1/3})$ edges and arboricity $O(\log n)$. For these graph classes, our algorithms provide exponentially faster solutions than the best known algorithms for general graphs in this model.
Beyond shortest paths, we also provide a variety of useful tools and techniques for hybrid networks, which may be of independent interest.
\end{abstract}

\section{Introduction} \label{sec:spa:intro}

The idea of \emph{hybrid networks} is to leverage multiple communication modes with different characteristics to deliver scalable throughput, or to reduce complexity, cost or power consumption.
In \emph{hybrid data center networks}~\cite{CGC16}, for example, the server racks can make use of optical switches \shrt{\cite{FPR+10}}\lng{\cite{FPR+10, WAK+10}} or wireless antennas \shrt{\cite{CWC11}}\lng{\cite{CWC11, CZL+13, HKPBW11, zzz+12}} to establish direct connections in addition to using the traditional electronic packet switches.
Other examples of hybrid communication are combining multipoint with standard VPN connections \cite{RS11}, hybrid WANs \cite{TBKC18}, or mobile phones using device-to-device communication in addition to cellular networks as in 5G \cite{KS18}.
As a consequence, several theoretical models and algorithms have been proposed for hybrid networks in recent years \cite{GHSS17, JKSS18, AGG+19, AHK+20}.

In this paper, we focus on the general hybrid network model of Augustine et al. \cite{AHK+20}.
The authors distinguish two different modes of communication, a \emph{local} mode, which nodes can use to send messages to their neighbors in an input graph $G$, and a \emph{global} mode, which allows the nodes to communicate with \emph{any} other node of $G$.
The model is parameterized by the number of messages $\lambda$ that can be sent over each local edge in each round, and the total number of messages $\gamma$ that each node can send and receive over global edges in a single round.
Therefore, the local network rather relates to \emph{physical} networks, where an edge corresponds to a dedicated connection that cannot be adapted by the nodes, e.g., a cable, an optical connection, or a wireless ad-hoc connection.
On the other hand, the global network captures characteristics of \emph{logical} networks, which are formed as overlays of a shared physical infrastructure such as the internet or a cellular network.
Here, nodes can in principle contact any other node, but can only perform a limited amount of communication in each round.

Specifically, we consider the hybrid network model with $\lambda = O(1)$ and $\gamma = O(\log n)$, i.e., the local network corresponds to the \textsf{CONGEST} model~\cite{Pel00}, whereas the global network is the so-called \emph{node-capacitated clique} (NCC) \cite{AGG+19, AAC+05, Rob20}.
Thereby, we only grant the nodes very limited communication capabilities for both communication modes, disallowing them, for example, to gather complete neighborhood information to support their computation.
With the exception of a constant factor SSSP approximation, none of the shortest paths algorithms of \cite{AHK+20}, for example, can be directly applied to this very restricted setting, since \cite{AHK+20} assumes the \textsf{LOCAL} model for the local network.
Furthermore, our algorithms do not even exploit the power of the NCC for the global network; in fact, they would also work if the nodes would initially only knew their neighbors in $G$ and had to learn new node identifiers via introduction (which has recently been termed the \emph{NCC$_0$} model~\cite{ACC+20}).

As in \cite{AHK+20}, we focus on \emph{shortest paths problems}.
However, instead of investigating general graphs, we present polylogarithmic time algorithms to compute Single-Source Shortest Paths (SSSP) and the diameter in \emph{sparse graphs}.
Specifically, we present randomized $O(\log n)$ time algorithms for \emph{cactus graphs}, which are graphs in which any two cycles share at most one node.
Cactus graphs are relevant for wireless communication networks, where they can model combinations of star/tree and ring networks (e.g., \cite{BDST12}), or combinations of ring and bus structures in LANs (e.g., \cite{LW00}).
However, research on solving graph problems in cactus graphs mostly focuses on the sequential setting\lng{ (e.g., \cite{BDST12, LW00, BH17, Das12, DP08, LWS99})}.

Furthermore, we present $3$-approximate randomized algorithms with runtime $O(\log^2 n)$ for graphs that contain at most $n + O(n^{1/3})$ edges and have arboricity\footnote{The arboricity of a graph $G$ is the minimum number of forests into which its edges can be partitioned.} $O(\log n)$.
Graphs with bounded arboricity, which include important graph families such as planar graphs, graphs with bounded treewidth, or graphs that exclude a fixed minor, have been extensively studied in the past years.
Note that although these graphs are very sparse, in contrast to cactus graphs they may still contain a polynomial number of (potentially nested) cycles.
Our algorithms are exponentially faster than the best known algorithms for general graphs for shortest paths problems~\cite{AGG+19, KS20}.

For the \emph{All-Pairs Shortest Paths} (APSP) problem, which is not studied in this paper, there is a lower bound of $\widetilde{\Omega}(\sqrt{n})$ \cite[Theorem 2.5]{AHK+20} that even holds for $\widetilde{O}(\sqrt{n})$-approximations\footnote{The $\widetilde{O}$-notation hides polylogarithmic factors.}. 
Recently, this lower bound was shown to be tight up to polylogarithmic factors~\cite{KS20}.
The bound specifically also holds for trees, which, together with the results in this paper, shows an exponential gap between computing the diameter and solving APSP in trees.
Furthermore, the results of \cite{KS20} show that computing (an approximation of) the diameter in general graphs takes time roughly $\Omega(n^{1/3})$ (even with unbounded local communication).
Therefore, our paper demonstrates that sparse graphs allow for an exponential improvement.



\subsection{Model and Problem Definition} \label{sec:spa:model}
We consider a \emph{hybrid network model} in which we are given a fixed node set $V$ consisting of $n$ nodes that are connected via \emph{local} and \emph{global} edges.
The local edges form a fixed, undirected, and weighted graph $G = (V,E,w)$ (the \emph{local network}), where the edge weights are given by $w: E \rightarrow \{1, \ldots, W\} \subset \mathbb{N}$ and $W$ is assumed to be polynomial in $n$.
We denote the degree of a node $v$ in the local network by $\deg(v)$.
Furthermore, every two nodes $u, v \in V$ are connected via a global edge, i.e., the \emph{global network} forms a clique.
Every node $v \in V$ has a unique identifier $\id(v)$ of size $O(\log n)$, and, since the nodes form a clique in the global network, every node knows the identifier of every other node.
Although this seems to be a fairly strong assumption, our algorithms would also work in the NCC$_0$ model \cite{ACC+20} for the global network, in which each node initially only knows the identifiers of its neighbors in $G$, and new connections need to be established by sending node identifiers (which is very similar to the overlay network models of \cite{GHSS17, AS18, GHS19}).
We further assume that the nodes know $n$ (or an upper bound polynomial in $n$).

We assume a synchronous message passing model, where in each round every node can send messages of size $O(\log n)$ over both local and global edges.
Messages that are sent in round $i$ are collectively received at the beginning of round $i+1$.
However, we impose different communication restrictions on the two network types.
Specifically, every node can send $O(1)$ (distinct) messages over each of its incident local edges, which corresponds to the \textsf{CONGEST} model for the local network \cite{Pel00}.
Additionally, it can send and receive at most $O(\log n)$ many messages over global edges (where, if more than $O(\log n)$ messages are sent to a node, an arbitrary subset of the messages is delivered), which corresponds to the NCC model \cite{AGG+19}.
Therefore, our hybrid network model is precisely the model proposed in \cite{AHK+20} for parameters $\lambda = O(1)$ and $\gamma = O(\log n)$.
Note that whereas \cite{AHK+20} focuses on the much more generous \textsf{LOCAL} model for the local network, our algorithms do not require nor easily benefit from the power of unbounded communication over local edges.

We define the \emph{length} of a path $P \subseteq E$ as $w(P) := \sum_{e \in P} w(e)$.
A path $P$ from $u$ to $v$ is a \emph{shortest path}, if there is no path $P'$ from $u$ and $v$ with $w(P') < w(P)$.
The \emph{distance} between two nodes $u$ and $v$ is defined as $d(u,v) := w(P)$, where $P$ is a shortest path from $u$ to $v$.

In the \emph{Single-Source Shortest Paths Problem} (SSSP), there is one node $s \in V$ and every node $v \in V$ wants to compute $d(s,v)$.
In the \emph{Diameter Problem}, every node wants to learn the \emph{diameter} $D := \max_{u,v \in V} d(u,v)$.
An algorithm computes an $\alpha$-approximation of SSSP, if every node $v \in V$ learns an estimate $\widetilde{d}(s,v)$ such that $d(s,v) \le \widetilde{d}(s,v) \le \alpha \cdot d(s,v)$.
Similarly, for an $\alpha$-approximation of the diameter, every node $v \in V$ has to compute an estimate $\widetilde{D}$ such that $D \le \widetilde{D} \le \alpha \cdot D$.

\subsection{Contribution and Structure of the Paper} \label{sec:spa:contribution}

The first part of the paper revolves around computing SSSP and the diameter on cactus graphs (i.e., connected graphs in which each edge is only contained in at most one cycle).
For a more comprehensive presentation, we establish the algorithm in several steps.
First, we consider the problems in path graphs (i.e., connected graphs that contain exactly two nodes with degree 1, and every other node has degree $2$; see Section~\ref{sec:spa:path}), then in cycle graphs (i.e., connected graphs in which each node has degree $2$, see Section~\ref{sec:spa:cycle}), trees (Section~\ref{sec:spa:tree}), and pseudotrees (Section~\ref{sec:spa:pseudotree}), which are graphs that contain at most one cycle.
For each of these graph classes, we present deterministic algorithms to solve both problems in $O(\log n)$ rounds, each relying heavily on the results of the previous sections.
We then extend our results to cactus graphs (Section~\ref{sec:spa:cactus_graph}) and present randomized algorithms for SSSP and the diameter with a runtime of $O(\log n)$, w.h.p.\footnote{An event holds with high probability (w.h.p.) if it holds with probability at least $1 - 1/n^c$ for an arbitrary but fixed constant c > 0.}

In Section~\ref{sec:spa:sparse_graph}, we consider a more general class of sparse graphs, namely graphs with at most $n + O(n^{1/3})$ edges and arboricity $O(\log n)$.
By using the techniques established in the first part and leveraging the power of the global network to deal with the additional $O(n^{1/3})$ edges, we obtain algorithms to compute $3$-approximations for SSSP and the diameter in time $O(\log^2 n)$, w.h.p.
As a byproduct, we also derive a deterministic $O(\log^2 n)$-round algorithm for computing a (balanced) hierarchical tree decomposition of the network.

We remark that our algorithms heavily use techniques from the PRAM literature.
For example, \emph{pointer jumping} \cite{JJ92}, and the \emph{Euler tour} technique (e.g., \cite{TV85, AV84}), which extends pointer jumping to certain graphs such as trees, have been known for decades, and are also used in distributed algorithms (e.g., \cite{GHSS17, AS18}).
As already pointed out in \cite{AGG+19}, the NCC in particular has a very close connection to PRAMs.
In fact, if $G$ is very sparse, PRAM algorithms can efficiently be simulated in our model even if the edges are very unevenly distributed (i.e., nodes have a very high degree).
We formally prove this in \shrt{the full version of this paper \cite{full}.}\lng{Appendix~\ref{sec:spa:pram_simulation}.}
This allows us to obtain some of our algorithms for path graphs, cycle graphs, and trees by PRAM simulations (see Section~\ref{sec:spa:related}).
We nonetheless present our distributed solutions without using PRAM simulations, since (1) a direct simulation \lng{as in Appendix~\ref{sec:spa:pram_simulation} }only yields randomized algorithms, (2) the algorithms of the later sections heavily build on the basic algorithms of the first sections, (3) a simulation exploits the capabilities of the global network more than necessary.
As already pointed out, \emph{all} of our algorithms would also work in the weaker NCC$_0$ model for the global network, or if the nodes could only contact $\Theta(\log n)$ random nodes in each round.\footnote{We remark that for the algorithms in Section~\ref{sec:spa:sparse_graph} this requires to setup a suitable overlay network like a butterfly in time $O(\log^2 n)$, which can be done using well-known techniques.}
Furthermore, if we restrict the degree of $G$ to be $O(\log n)$, our algorithms can be modified to run in the NCC$_0$ without using the local network.

Beyond the results for sparse graphs, this paper contains a variety of useful tools and results for hybrid networks in general, such as Euler tour and pointer jumping techniques for computation in trees, a simple load-balancing framework for low-arboricity graphs, an extension of the recent result of Götte et al. \cite{GHSW20} to compute spanning trees in the NCC$_0$, and a technique to perform matrix multiplication.
In combination with sparse spanner constructions (see, e.g., \cite{BS07}) or skeletons (e.g., \cite{UY91}), our algorithms may lead to efficient shortest path algorithms in more general graph classes.
Also, our algorithm to construct a hierarchical tree decomposition may be of independent interest, as such constructions are used for example in routing algorithms for wireless networks (see, e.g., \cite{GZ05, KMRS18}).

\shrt{Due to space constraints, all proofs and figures, as well as the detailed description and some lemmas of our algorithms, are deferred to the full version of this paper~\cite{full}.}

\subsection{Further Related Work} \label{sec:spa:related}
As theoretical models for hybrid networks have only been proposed recently, only few results for such models are known at this point \cite{GHSS17,AGG+19,AHK+20}.
Computing an exact solution for SSSP in arbitrary graphs can be done in $\widetilde{O}(\sqrt{\textsf{SPD}})$ rounds~\cite{AHK+20}, where $\textsf{SPD}$ is the so-called \emph{shortest path diameter} of $G$.
For large $\textsf{SPD}$, this bound has recently been improved to $\widetilde{O}(n^{2/5})$~\cite{KS20}.
The authors of~\cite{AHK+20} also present several approximation algorithms for SSSP: A $(1+\varepsilon)$-approximation with runtime $\widetilde{O}(n^{1/3}/\varepsilon^6)$, a $(1/\varepsilon)^{O(1/\varepsilon)}$-approximation running in $\widetilde{O}(n^{\varepsilon})$ rounds and a $2^{O(\sqrt{\log n \log \log n})}$-approximation with runtime $2^{O(\sqrt{\log n \log \log n})}$.
For APSP there is an exact algorithm that runs in $\widetilde{O}(n^{2/3})$ rounds, a $(1+\varepsilon)$-approximation running in $\widetilde{O}(\sqrt{n/\varepsilon})$ rounds (only for unweighted graphs) and a $3$-approximation with runtime $\widetilde{O}(\sqrt{n})$~\cite{AHK+20}.
In \cite{KS20}, the authors give a lower bound of $\widetilde{\Omega}(n^{1/3})$ rounds for computing the diameter in arbitrary graphs in our model.
They also give approximation algorithms with approximation factors $(3/2 + \varepsilon)$ and $(1+\varepsilon)$ that run in time $\widetilde{O}(n^{1/3}/\varepsilon)$ and $\widetilde{O}(n^{0.397}/\varepsilon)$, respectively.
Even though APSP and the diameter problem are closely related, we demonstrate that the diameter can be computed much faster in our hybrid network model for certain graphs classes.

As already pointed out, the global network in our model has a close connection to overlay networks.
The NCC model, which has been introduced in \cite{AGG+19}, mainly focuses on the impact of node capacities, especially when the nodes have a high degree.
Since, intuitively, for many graph problems the existence of \emph{each} edge is relevant for the output, most algorithms in \cite{AGG+19} depend on the arboricity $a$ of $G$ (which is, roughly speaking, the time needed to efficiently distribute the load of all edges over the network).
The authors present $\widetilde{O}(a)$ algorithms for local problems such as MIS, matching, or coloring, an $\widetilde{O}(D+a)$ algorithm for BFS tree, and an $\widetilde{O}(1)$ algorithm to compute a minimum spanning tree (MST).
Recently, $\widetilde{O}(\Delta)$-time algorithms for graph realization problems have been presented \cite{ACC+20}, where $\Delta$ is the maximum node degree; notably, most of the algorithms work in the NCC$_0$ variant.
Furthermore, Robinson \cite{Rob20} investigates the information the nodes need to learn to jointly solve graph problems and derives a lower bound for constructing spanners in the NCC.
For example, his result implies that spanners with constant stretch require polynomial time in the NCC, and are therefore harder to compute than MSTs.
Since our global network behaves like an overlay network, we can make efficient use of the so-called \emph{shortest-path diameter reduction technique} \cite{Nan14}.
By adding shortcuts between nodes in the global network, we can bridge large distances quickly throughout our computations.

\lng{
Our work also relates to the literature concerned with \emph{overlay construction} \cite{AW07, AAC+05, GHSS17, AS18, GHS19, GHSW20}, where the goal is to transform a low-degree weakly-connected graph into a low-depth overlay network such as a binary tree using node introductions.
Recently, \cite{GHSW20} showed a randomized $O(\log n)$ time overlay construction algorithm for the NCC$_0$ model, if the initial degree is constant.
Our algorithms directly yield a \emph{deterministic} $O(\log n)$ time alternative for pseudotrees.
It may be interesting to see whether our results can be used to compute deterministic overlay construction algorithms for sparse graphs.

A problem closely related to SSSP is the computation of short routing paths between any given nodes.
The problem has, for example, been studied in mobile ad-hoc networks \cite{JKSS18}, in which constant-competitive routing paths can be computed in $O(\log^2 n)$ rounds \cite{CKS19}.
The authors consider a hybrid network model similar to \cite{AHK+20}, where nodes can communicate using either their WiFi-interface (similar to the local edges) or the cellular infrastructure (similar to global edges).

In the classical \textsf{CONGEST} model there is a lower bound of $\widetilde{\Omega}(\sqrt{n} + D)$ rounds to approximate SSSP with a constant factor~\cite{SHKKNPPW12}.
This bound is tight, as there is a $(1+\varepsilon)$-approximation algorithm by Becker et al. that runs in $\widetilde{O}(\sqrt{n} + D)$ rounds~\cite{BKKL17}.
The best known algorithms for computing exact SSSP in the \textsf{CONGEST} model are the ones by Ghaffari and Li~\cite{GL18} and by Forster and Nanongkai~\cite{FN18} which have runtimes of $\widetilde{O}(\sqrt{n \cdot D})$ and $\widetilde{O}(\sqrt{n} D^{1/4} + n^{3/5} + D)$, respectively.
Computing the diameter can be done in $O(n)$ rounds in the \textsf{CONGEST} model~\cite{PRT12}, which is also tight to the lower bound~\cite{FHW12}.
This lower bound even holds for very sparse graphs \cite{ACK16}.
In addition to that, the obvious lower bound of $\Omega(D)$ for shortest paths problems also always holds if the graph is sparse.
Therefore, algorithms for sparse graphs have been proposed mainly for \emph{local} problems such as vertex coloring, maximal matching or maximal independent set.
There exists an abundance of literature that studies such problems, for example, in graphs with bounded arboricity~\cite{GhaffariL17,BE10,KothapalliP11}, planar graphs~\cite{AboulkerBBE18,ChechikM19,DBLP:conf/podc/GhaffariH16, DBLP:conf/soda/GhaffariH16} or degree-bounded graphs~\cite{PanconesiR01}.

Somewhat related to the NCC model, although much more powerful, is the \emph{congested clique} model, which has received quite some attention in recent years.
A $(1+\varepsilon)$-approximation for SSSP can be computed in $O(\text{polylog}(n))$ rounds in this model~\cite{Censor_Hillel_2019}.
In~\cite{CKK19}, techniques for faster matrix multiplication in the congested clique model are presented, resulting in a $O(n^{1-2/\omega})$-round algorithm, where $\omega < 2.3728639$ is the exponent of matrix multiplication.
Our algorithm for sparse graphs also uses matrix multiplication in order to compute APSP between $O(n^{1/3})$ nodes in the network in $O(\log^2 n)$ rounds.
In general, the results in the congested clique model are of no help in our setting because due to the restriction that a node can only send or receive $O(\log n)$ messages per round via global edges, we cannot effectively emulate congested clique algorithms in the NCC model.
}

As argued before, we could apply some of the algorithms for PRAMs to our model instead of using native distributed solutions by using \lng{Lemma~\ref{lem:spa:pram_sim} in Appendix~\ref{sec:spa:pram_simulation} for }PRAM simulations.
For example, we are able to use the algorithms of \cite{DPZ91} to solve SSSP and diameter in trees in time $O(\log n)$, w.h.p.
Furthermore, we can compute the distance between any pair $s$ and $t$ in \emph{outerplanar graphs} in time $O(\log^3 n)$ by simulating a CREW PRAM.
For planar graphs, the distance between $s$ and $t$ can be computed in time $O(\log^3 n (1 + M(q))/n)$, w.h.p., where the nodes know a set of $q$ faces of a planar embedding that covers all vertices, and $M(q)$ is the number of processors required to multiply two $q \times q$ matrices in $O(\log q)$ time in the CREW PRAM.

For graphs with polylogarithmic arboricity, a $(1+\varepsilon)$-approximation of SSSP can be computed in polylog time using \cite{Li20} and our simulation framework (with huge polylogarithmic terms).
For general graphs, the algorithm can be combined with well-known spanner algorithms for the \textsf{CONGEST} model (e.g., \cite{BS07}) to achieve constant approximations for SSSP in time $\widetilde{O}(n^\varepsilon)$ time in our hybrid model. 
This yields an alternative to the SSSP approximation of \cite{AHK+20}, which also requires time $\widetilde{O}(n^\varepsilon)$ but has much smaller polylogarithmic factors.

\section{Path Graphs} \label{sec:spa:path}

To begin with an easy example, we first present a simple algorithm to compute SSSP and the diameter of path graphs.
The simple idea of our algorithms is to use \emph{pointer jumping} to select a subset of global edges $S$, which we call \emph{shortcut edges}, with the following properties: 
$S$ is a weighted connected graph with degree $O(\log n)$ that contains all nodes of $V$, and for every $u, v \in V$ there exists a path $P \subseteq S$, $|P| = O(\log n)$ (where $|P|$ denotes the number of edges of $P$), such that $w(P) = d(u,v)$, and no path $P$ such that $w(P) < d(u,v)$.
Given such a graph, SSSP can easily be solved by performing a broadcast from $s$ in $S$ for $O(\log n)$ rounds:
In the first round, $s$ sends a message containing $w(e)$ over each edge $e \in S$ incident to $s$.
In every subsequent round, every node $v \in V$ that has already received a message sends a message $k + w(e)$ over each edge $e \in S$ incident to $v$, where $k$ is the smallest value $v$ has received so far.
After $O(\log n)$ rounds, every node $v$ must have received $d(s, v)$, and cannot have received any smaller value.
Further, the diameter of the line can easily be determined by performing SSSP from both of its endpoints $u,v$, which finally broadcast the diameter $d(u,v)$ to all nodes using the global network.


We construct $S$ using the following simple \emph{Introduction Algorithm}.
$S$ initially contains all edges of $E$.
Additional shortcut edges are established by performing \emph{pointer jumping}:
Every node $v$ first selects one of its at most two neighbors as its \emph{left} neighbor $\ell_1$; if it has two neighbors, the other is selected as $v$'s \emph{right} neighbor $r_1$.
In the first round of our algorithm, every node $v$ with degree $2$ establishes $\{\ell_1, r_1\}$ as a new shortcut edge of weight $w(\{\ell_1, r_1\}) = w(\{\ell_1, v\}) + w(\{v, r_1\})$ by sending the edge to both $\ell_1$ and $r_1$.
Whenever at the beginning of some round $i>1$ a node $v$ with degree $2$ receives shortcut edges $\{u, v\}$ and $\{v, w\}$ from $\ell_{i-1}$ and $r_{i-1}$, respectively, it sets $\ell_i := u$, $r_i := w$, and establishes $\{\ell_{i}, r_i\}$ by adding up the weights of the two received edges and informing $\ell_{i}$ and $r_i$.
The algorithm terminates after $\lfloor \log (n - 1) \rfloor$ rounds.
Afterwards, for every simple path in $G$ between $u$ and $v$ with $2^k$ hops for any $k \le \lfloor \log (n - 1) \rfloor$ we have established a shortcut edge $e \in S$ with $w(e) = d(u,v)$.
Therefore, $S$ has the desired properties, and we conclude the following theorem.

\begin{theorem}\label{thm:spa:path}
    SSSP and the diameter can be computed in any path graph in time $O(\log n)$.
\end{theorem}

\section{Cycle Graphs} \label{sec:spa:cycle}

In cycle graphs, there are two paths between any two nodes that we need to distinguish.
For SSSP, this can easily be achieved by performing the SSSP algorithm for path graphs in both directions along the cycle, and let each node choose the minimum of its two computed distances.
Formally, let $v_1, v_2, \ldots, v_n$ denote the $n$ nodes along a \emph{left} traversal of the cycle starting from $s = v_1$ and continuing at $s$'s neighbor of smaller identifier, i.e., $\id(v_2) < \id(v_n)$.
For any node $u$, a shortest path from $s$ to $u$ must follow a left or right traversal along the cycle, i.e., $(v_1, v_2, \ldots, u)$ or $(v_1, v_n, \ldots, u)$ is a shortest path from $s$ to $u$.
Therefore, we can solve SSSP on the cycle by performing the SSSP algorithm for the path graph on $\mathcal{L} := (v_1, v_2, \ldots, v_n)$ and $\mathcal{R} := (v_1, v_n, v_{n-1}, \ldots, v_2)$.
Thereby, every node $v$ learns $d_\ell(s,v)$, which is the distance from $s$ to $v$ in $\mathcal{L}$ (i.e., along a left traversal of the cycle), and $d_r(s,v)$, which is their distance in $\mathcal{R}$.
It is easy to see that $d(s,v) = \min\{d_\ell(s,v), d_r(s,v)\}$.

Using the above algorithm, $s$ can also easily learn its \emph{eccentricity} $\ecc(s) := \max_{v \in V}\{d(s,v)\}$, as well as its \emph{left and right farthest nodes} $s_\ell$ and $s_r$.
The left farthest node $s_\ell$ of $s$ is defined as the farthest node $v_i$ along a left traversal of the cycle such that the subpath in $\mathcal{L}$ from $s = v_1$ to $v_i$ is still a shortest path.
Formally, $s_\ell = \argmax_{v \in V, d_\ell(s, v) \le \lfloor W/2 \rfloor} d_\ell(s, v)$, where $W = \sum_{e \in E}w(e)$.
The right farthest node $s_r$ is the successor of $s_\ell$ in $\mathcal{L}$ (or $s$, if $s_\ell$ is the last node of $\mathcal{L})$, for which it must hold that $d_r(s,s_r) \le \lfloor W/2 \rfloor$.
Note that $d_\ell(s,s_\ell) = d(s,s_\ell)$, $d_r(s,s_r) = d(s,s_r)$, and $\ecc(s) = \max\{d_\ell(s,s_\ell), d_r(s,s_r)\}$.

To determine the diameter of $G$, for every node $v \in V$ our goal is to compute $\ecc(v)$; as a byproduct, we will compute $v$'s left and right farthest nodes $v_\ell$ and $v_r$.
The diameter can then be computed as $\max_{v \in V} \ecc(v)$.
A simple way to compute these values is to employ a binary-search style approach from all nodes in parallel, and use load balancing techniques from \cite{AGG+19} to achieve a runtime of $O(\log^2 n)$, w.h.p.
Coming up with a deterministic $O(\log n)$ time algorithm, however, is more complicated.

\shrt{Due to space constraints, we defer the description of the algorithm to the full version.}

\lng{
\begin{figure}[t]
    \centering
    \includegraphics{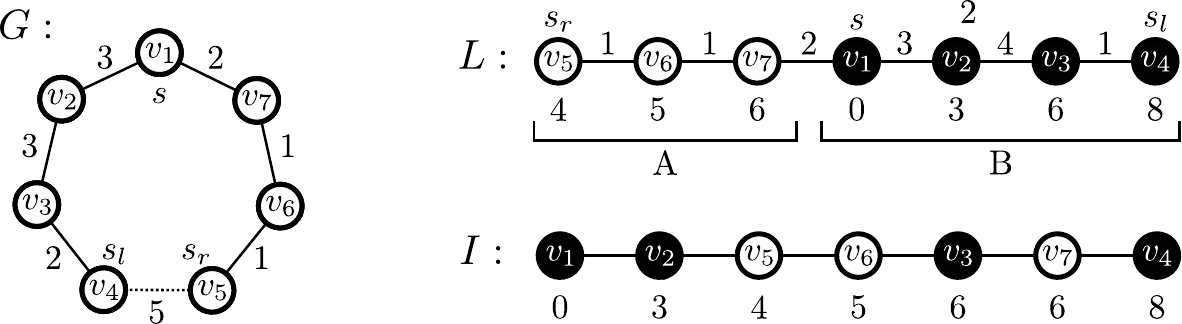}
    \caption{An example of diameter computation in a cycle $G$.
    The algorithm begins with $s = v_1$.
    In $L$, $s_{\ell} = v_4$ is the farthest node from $s$ along a left traversal of $G$, and $s_r = v_5$ is the farthest node along a right traversal.
    The white nodes are the nodes of $A$, and the black nodes are $B$.
    Each node is annotated with its budget.
    In $I$, the nodes are sorted by their budget, and learn their nearest black nodes.
    For example, for $v_3$, $x=v_1$ and $y = v_7$.
    }
    \label{fig:diam_cycle}
\end{figure}

Our algorithm works as follows.
Let $s$ be the node with highest identifier\footnote{In the NCC$_0$, this node can easily be determined by performing pointer jumping in the cycle.}.
First, we perform the SSSP algorithm as described above from $s$ in $\mathcal{L}$ and $\mathcal{R}$, whereby $s$ learns $s_\ell$ and $s_r$ as defined above.
Let $L$ be graph that results from removing the edge $\{s_\ell, s_r\}$ from $G$ \lng{(see Figure~\ref{fig:diam_cycle}) }.
Let $A \subseteq V$ be the set of nodes between $s_r$ and $s$ (excluding $s$), and $B \subseteq V$ be the set of nodes between $s$ and $s_\ell$ (including $s$).

In its first execution, our algorithm ensures that each node $v \in A$ learns its left farthest node $v_\ell$; a second execution will then handle all other nodes.
Note that $v_\ell$ for all $v \in A$ must be a node of $B$, since otherwise the path from $v$ to $v_\ell$ along $\mathcal{L}$ is longer than $\lfloor W/2 \rfloor$, in which case it cannot be a shortest path anymore.

We assign each node $v$ a budget $\phi(v)$, which is $\lfloor W /2 \rfloor - d_r(s,v) \ge 0$, if $v \in A$, and $d_\ell(s, v)$, if $v \in B$.
Roughly speaking, the budget of a node $v \in A$ determines how far you can move from $v$ beyond $s$ along a left traversal of $G$ until reaching $v$'s left farthest node $v_\ell$.
Then, we sort the nodes of $L$ by their budget.
Note that since we consider positive edge weights, no two nodes of $A$ and no two nodes of $B$ may have the same budget, but there may be nodes $u \in A$, $v \in B$ with $\phi(u) = \phi(v)$.
In this case, we break ties by assuming that $\phi(u) > \phi(v)$.
More specifically, the outcome is a sorted list $I = (s = v_{i_1}, v_{i_2}, \ldots, v_{i_n})$ with first node $s$ that contains all nodes of $A$ (and $B$) in the same order they appear in $L$, respectively.
Such a list can be constructed in time $O(\log n)$, e.g., by using Aspnes and Wu's algorithm \cite{AW07}.\footnote{Note that the algorithm of \cite{AW07} is actually randomized. 
However, since we can easily arrange the nodes as a binary tree, we can replace the randomized pairing procedure of \cite{AAC+05} by a deterministic strategy, and, together with the pipelining approach of \cite{AW07}, also achieve a runtime of $O(\log n)$ .}

Let $v = v_{i_k} \in A$, and let $x = v_{i_j} \in B$ be the node with maximum index $j$ in $I$ such that $j < k$ (i.e., the last node of $B$ in $I$ that is still before $v$).
Since the nodes in $I$ are sorted by their potential, among all nodes of $B$, $x$ maximizes $\phi(x)$ such that $\phi(x) \le \phi(v)$.
By definition of $\phi(x)$ and $\phi(v)$, this implies that 
\[
    j = \max\{j \in \{1, \ldots, n\} \mid v_{i_j} \in B, j < k \text{ and } d_r(s,v_{i_k}) + d_\ell(s,v_{i_j}) \le \lfloor W /2 \rfloor\}.
\]

\begin{lemma} \label{lem:spa:circ_dia_1}
    We have $x = v_\ell$.
\end{lemma}

\begin{proof}
    By the definition of our algorithm, $x\in B$ is the farthest node from $v\in A$ along a left traversal of the cycle such that $d_r(s,v) + d_\ell(s,x) \le \lfloor W /2 \rfloor$.
    Note that $d_r(s,v) + d_\ell(s,x) = d_\ell(s,v)$, since $s$ lies between $v$ and $x$ in $L$.
    Therefore, $x$ is also farthest from $v$ along a left traversal such that $d_\ell(v,x) \le \lfloor W /2 \rfloor$, which is the definition of $v_\ell$.
\end{proof}

Node $v$ can easily learn $x$, $d_\ell(s,x)$, and the neighbors of $x$ in the cycle (to infer $v_r$) by performing the Introduction Algorithm on each connected segment of nodes of $A$ in $I$.
To let all remaining nodes learn their farthest nodes, we restart the algorithm at node $s_\ell$ (instead of $s$).
Since $d_\ell(s,s_\ell) = d_r(s_\ell,s) \le \lfloor W/2 \rfloor$, all nodes between $s$ and $s_\ell$ in $L$ (except $s_\ell$), which previously were in set $B$, will be in set $A$ and learn their farthest nodes.
Finally, $s_r$ learns its farthest nodes by performing SSSP.
We conclude the following theorem.
}

\begin{theorem} \label{thm:spa:sssp:cycle}
    SSSP and the diameter can be computed in any cycle graph $G$ in time $O(\log n)$.
\end{theorem}

\section{Trees} \label{sec:spa:tree}

We now show how the algorithms of the previous sections can be extended to compute SSSP and the diameter on trees.
As in the algorithm of Gmyr et al.~\cite{GHSS17}, we adapt the well-known \emph{Euler tour} technique to a distributed setting and transform the graph into a path $L$ of \emph{virtual nodes} that corresponds to a depth-first traversal of $G$.
More specifically, every node of $G$ simulates one virtual node for each time it is visited in that traversal, and two virtual nodes are neighbors in $L$ if they correspond to subsequent visitations.
To solve SSSP, we assign weights to the edges from which the initial distances in $G$ can be inferred, and then solve SSSP in $L$ instead.
Finally, we compute the diameter of $G$ by performing the SSSP algorithm twice, which concludes this section.

However, since a node can be visited up to $\Omega(n)$ times in the traversal, it may not be able to simulate all of its virtual nodes in $L$.
Therefore, we first need to reassign the virtual nodes to the node's neighbors such that every node only has to simulate at most 6 virtual nodes using the \emph{Nash-Williams forests decomposition} technique~\cite{Nas64}. 
More precisely, we compute an \emph{orientation} of the edges in which each node has outdegree at most $3$, and reassign nodes according to this orientation (in the remainder of this paper, we refer to this as the \emph{redistribution framework}).
\shrt{Due to space constraints, we defer a precise description of the algorithm to the full version and only state our main results.

The following two lemmas follow from applying PRAM techniques.}

\lng{

\begin{figure}[t]
\centering
\begin{subfigure}{69pt}
	\centering
	\includegraphics{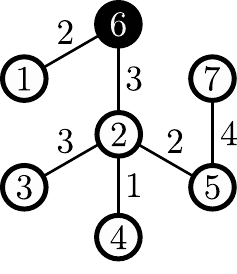}
	\subcaption{\centering}
	\label{fig:tree_1}
\end{subfigure}
\hfill
\begin{subfigure}{71pt}
	\centering
	\includegraphics{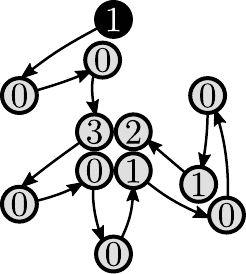}
	\subcaption{\centering}
	\label{fig:tree_2}
\end{subfigure}
\hfill
\begin{subfigure}{72pt}
	\centering
	\includegraphics{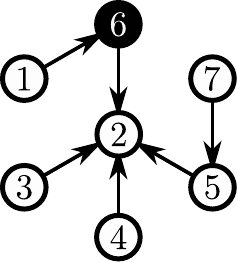}
	\subcaption{\centering}
	\label{fig:tree_4}
\end{subfigure}
\hfill
\begin{subfigure}{72pt}
	\centering
	\includegraphics{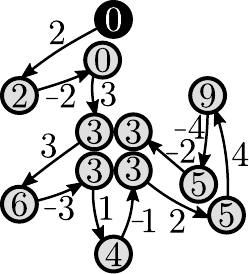}
	\subcaption{\centering}
	\label{fig:tree_3}
\end{subfigure}
\caption{
(a) A tree with source node $s$ (black).
Each node is labeled with its identifier and each edge is labeled with its weight.
(b) The resulting path graph $L$ of virtual nodes.
Each node $v_i$ is labeled with its index $i$.
(c) A possible orientation with outdegree $3$.
According to our redistribution rule, for example, all virtual nodes of the central node $2$ would be assigned to its neighbors.
(d) The edges are assigned weights, and each virtual node is labeled with its distance to $s_L$ (black node).
}
\label{fig:tree}
\end{figure}

\textbf{Construction and Simulation of $L$.}
We denote the neighbors of a node $v \in V$ by ascending identifier as $v(0), \ldots, v(\deg(v) - 1)$.
Consider the depth-first traversal in $G$ that starts and ends at $s$, and which, whenever it reaches $v$ from some neighbor $v(i)$, continues at $v$'s neighbor $v((i+1) \bmod \deg(v))$.
$L$ is the directed path graph of virtual nodes that corresponds to this traversal (see Figure~\ref{fig:tree_1} and \ref{fig:tree_2}).
The path graph contains a virtual node for each time a node is visited, and a directed edge from each virtual node to its successor in the traversal; however, we leave out the last edge ending at $s$ to break the cycle.
More specifically, every node $v$ simulates the nodes $v_0, \ldots, v_{\text{deg}(v) - 1}$, where $v_i$ corresponds to the traversal visiting $v$ from $v(i)$.
The first node of $L$ is $s_L := s_{\deg(s) - 1}$, and its last node is the node $v_i$ such that $v = s(\deg(s) - 1)$, and $v((i+1) \bmod \deg(v)) = s$.
For every node $v_i$ in $L$ (except the last node of $L$), there is an edge $(v_i, u_j) \in L$ such that $u = v((i+1) \bmod \deg(v))$ and $v = u(j)$.
To accordingly introduce each virtual node to its predecessor in $L$, every node $v$ sends the \emph{virtual identifiers} $\id(v_i) := \id(v) \circ i$ to $v(i)$ for all $i \in [\deg(v)]$, where $\circ$ denotes the concatenation of two binary strings, and $[k] = \{0, \ldots, k-1\}$.

It remains to show how the virtual nodes can be redistributed such that each node only has to simulate at most $6$ virtual nodes.
To do so, we first compute an orientation of $G$, i.e., an assignment of directions to its edges, such that every node has outdegree $3$.

Since the arboricity of $G$ is 1, we can use \cite[Theorem 3.5]{BE10} to compute an \emph{$H$-partition $H_1, \ldots, H_\ell$ of $G$ with degree $3$}.
The algorithm is based on the \emph{Nash-Williams forests decomposition} technique~\cite{Nas64}:
In phase $i \in \{1, \ldots, \ell = O(\log n)\}$, all nodes that have degree at most $(2 + \varepsilon) \cdot a$, where $a$ is the arboricity of $G$, are removed from the graph and join the set $H_i$.
We obtain our desired orientation by directing each edge $\{u,v\} \in E$, $u \in H_i$, $v \in H_j$, from $u$ to $v$ if $i < j$, or $i = j$ and $\id(u) < \id(v)$ (see Figure~\ref{fig:tree_4} for an example).

Now consider some node $v \in V$ and a virtual node $v_i$ at $v$, and let $u := v(i)$.
If $\{v, u\}$ is directed from $v$ to $u$, then $v_i$ is assigned to $v$, and, as before, $v$ takes care of simulating $v_i$.
Otherwise, $v_i$ gets assigned to $u$ instead, and $v$ sends the identifier of $v_i$ to $u$.
Afterwards, $u$ needs to inform the node $w$ that is responsible for simulating the predecessor of $v_i$ in $L$ that the location of $v_i$ has changed; as $w$ must be either $u$ itself, or a neighbor of $u$, this can be done in a single round.
Since in the orientation each node $v$ has at most $3$ outgoing edges, for each of which it keeps one virtual node and is assigned one additional virtual node from a neighbor, $v$ has to simulate at most $6$ virtual nodes of $L$.

As a byproduct, we obtain that if $G$ is any forest, we can establish separate low-diameter overlays on each of its trees by combining the techniques used in this section and the pointer jumping approach of Section~\ref{sec:spa:path}.
For instance, this allows us to efficiently compute aggregates of values stored at each tree's nodes, as stated in the following lemma.

}

\begin{lemma}\label{lem:spa:aggregation}
    Let $H = (V,E)$ be a forest in which every node $v \in V$ stores some value $p_v$, and let $f$ be a distributive aggregate function\footnote{An aggregate function $f$ is called \emph{distributive} if there is an aggregate function $g$ such that for any multiset $S$ and any partition $S_1,\ldots,S_\ell$ of $S$, $f(S)=g(f(S_1),\ldots,f(S_\ell))$.
    Classical examples are MAX, MIN, and SUM.}.
    Every node $v \in V$ can learn $f(\{p_u \mid u \in C_v\})$, where $C_v$ is the tree of $H$ that contains $v$, in time $O(\log n)$.
\end{lemma}

\lng{ 
\textbf{Assigning Weights.}
To assign appropriate weights to the edges of $L$ from which we can infer the node's distances in $G$, we first have to transform $G$ into a rooted tree.
To do so, we simply perform SSSP from $s_L$ (the first node in $L$) in the (unweighted) version of $L$.
Thereby, every virtual node $x$ learns its \emph{traversal distance}, i.e., how many steps the depth-first traversal takes until it reaches $x$.
Further, every node $v$ can easily compute which of its virtual nodes $v_i$ is visited first by taking the minimum traversal distance of its virtual nodes\footnote{Since the virtual nodes of $v$ might actually be assigned to neighbors of $v$, their traversal distances first have to be sent to $v$ using the local network.}.
Let $v_i$ be the virtual node of $v$ that has smallest traversal distance, and let $u_j$ be the predecessor of $v_i$ in $L$.
It is easy to see that $u$ is the parent of $v$ in the rooted tree, which implies the following lemma.

} 
\begin{lemma}\label{lem:spa:tree_rooting}
    Any tree $G$ can be rooted in $O(\log n)$ time.
\end{lemma}

\lng{ 
For each virtual node $v_j$ of $v$ (except the first node $s_L$), to the edge $(u_i, v_j) \in L$, $v$ assigns the weight 
\shrt{
$w(u_i, v_j) = w(\{u,v\})$, if $u$ is $v$'s parent, and $-w(\{u,v\})$, otherwise.}
\lng{\[
 w(u_i, v_j) = \begin{dcases*}
        w(\{u,v\})  & if $u$ is $v$'s parent\\
        -w(\{u,v\}) & if $v$ is $u$'s parent.
        \end{dcases*}
\]}
If $v_j$ is assigned to a neighbor of $v$, it informs that neighbor about the weight (see Figure~\ref{fig:tree_3}).

To solve SSSP in $G$, we simply compute SSSP in $L$ using Theorem~\ref{thm:spa:path}\footnote{Note that for the algorithm to work in the \emph{directed} path graph $L$, shortcuts must be established in the \emph{bidirected} version of $L$, whereas the subsequent broadcast from $s_L$ uses only the directed edges of $L$.}.
As we prove in the following theorem, the distance of each virtual node $v_i$ of each node $v$ will be $d(s,v)$.
}

\shrt{By assigning positive or negative weights to the edges of $L$ according to their direction in the rooted version of $G$, we easily obtain the following theorem.}
\begin{theorem} \label{thm:spa:sssp:tree}
    SSSP can be computed in any tree in time $O(\log n)$.
\end{theorem}

\lng{
\begin{proof}
    Let $v \in V$, and let $d_L(s_L, v_i)$ denote the distance from $s_L$ to a virtual node $v_i$ at $v$ in the (weighted) graph $L$.
    We show that $d_L(s_L,v_i) = d(s,v)$.
    
    Consider the path $P$ from $s$ to $v$ in $G$.
    The depth-first traversal from $s$ to $v$ traverses every edge of $P$ from parent to child, i.e., for every edge in $P$ there is a directed edge with the same weight between $s$ and $v_i$ in $L$.
    However, at some of the nodes of $P$ (including $s$ and $v$) the traversal may take \emph{detours} into other subtrees before traversing the next edge of $P$.
    As every edge of $L$ that corresponds to an edge in the subtree is visited, and the weights of all those edges sum up to $0$, the distance from $s$ to $v_i$ equals the sum of all edges in $P$, which is $d(s,v)$.
\end{proof}
}
Similar techniques lead to the following lemmas, which we will use in later sections.

\begin{lemma} \label{lem:spa:subtree_size_aggregate}
    Let $H = (V,E)$ be a forest and assume that each node $v \in V$ stores some value $p_v$.
    The goal of each node $v$ is to compute the value $\sumT_v(u) := \sum_{w \in C_u} p_w$ for each of its neighbors $u$, where $C_u$ is the connected component $C$ of the subgraph $H'$ of $H$ induced by $V \setminus \{v\}$ that contains $u$.
    The problem can be solved in time $O(\log n)$.
\end{lemma}

\lng{
\begin{proof}
    Let $s \in V$ be the node that has highest identifier in $V$, which can easily be computed using Lemma~\ref{lem:spa:aggregation}.
    We construct $L$ exactly as described in the algorithm for computing SSSP with source $s$ on trees, but choose the weights of the edges differently.
    More precisely, to every edge $(u_i, v_j)$ of $L$ we assign the weight $w(\{u_i,v_j\}) := p_u$, if $v$ is $u$'s parent, and $0$, otherwise.
    Further, we assign a value $\hat{d}(s_L) := p_{s}$ to $s_L$ (the first node of $L$).
    With these values as edge weights, we perform the SSSP algorithm on $L$ from $s_L$, whereby every virtual node $v_i$ learns the value $\hat{d}(v_i) := \hat{d}(s_L) + d_L(s_L, v_i)$.
    The sum of all values $M := \sum_{v \in V} p_v$ can be computed and broadcast to every node of $H$ in time $O(\log n)$ using Lemma~\ref{lem:spa:aggregation}.

    The problem can now be solved as follows.
    Consider some node $v$, let $u$ be a neighbor of $v$, and let $i$ be the value such that $u = v(i)$ (recall that $v(i)$ is the neighbor of $v$ that has the $i$-th highest identifier, $0 \le i \le \deg(v) - 1$).
    If $u$ is the parent of $v$ in the tree rooted at $s$, then $\sum_{w \in C_u} p_w = M - (\hat{d}(v_{i-1 \bmod \deg{v}}) - \hat{d}(v_{i}))$.
    If otherwise $u$ is a child of $v$ in the tree rooted at $s$ (unless $v = s$ and $i = \deg(s) - 1$, which is a special case), then $\sum_{w \in C_u} p_w = \hat{d}(v_{i}) - \hat{d}(v_{i-1 \bmod \deg{v}})$.
    Finally, if $v = s$ and $i = \deg(s) - 1$, we have $\sum_{w \in C_u} p_w = M - \hat{d}(v_{i-1 \bmod \deg{v}})$.
\end{proof}
}

\begin{lemma} \label{lem:spa:height}
    Let $G$ be a tree rooted at $s$.
    Every node $v \in V$ can compute its height $h(v)$ in $G$, which is length of the longest path from $v$ to any leaf in its subtree, in time $O(\log n)$.
\end{lemma}

\lng{
\begin{proof}
    By Theorem~\ref{thm:spa:sssp:tree}, each leaf node $v$ can learn its distance $d(s,v)$ to $s$ (its \emph{depth} in the tree) in time $O(\log n)$.
    For $v = s$, the height is the maximum depth of any node, which can be computed by performing one aggregation using Lemma~\ref{lem:spa:aggregation}.
    For any other node $v \neq s$, the height $h(v)$ is the maximum depth of any leaf in its subtree minus the depth $d(s,v)$ of $v$.
    To allow each node to compute the maximum depth within its subtree, every leaf node $u$ assigns its virtual node in $L$ the value $d(s,u)$
    We then again establish shortcuts on $L$ using the Introduction Algorithm; however, we begin with each edge having the maximum value assigned to its endpoints.
    Whenever a shortcut results from two smaller shortcuts being merged, its weight becomes the maximum weight of the two smaller shortcuts.
    Thereby, the weight of each shortcut $\{u_i, v_j\}$ (where $u_i$ is the endpoint that is closer to $s_L$) corresponds to the maximum value of any node in $G$ visited by the traversal from $u_i$ to $v_j$.
    
    Slightly abusing notation, let $x_0 = s_L, x_2, \ldots, x_{2(n-1)}$ denote all virtual nodes in the order they appear in $L$ (note that the index of each virtual node is its traversal distance, which can easily be computed).
    Now let $v \neq s$, and $x_i$ be the virtual node of $v$ with smallest, and $x_j$ be its virtual node with highest traversal distance.
    Let $k = 2^{\lfloor \log (j -i) \rfloor}$.
    Note that the shortcuts $\{x_i, x_{i+k}\}$ and $\{x_{j-k}, x_j\}$ exist and overlap, and therefore span all virtual nodes of the nodes in the subtree of $v$ in $G$.
    Therefore, the value $\max\{w(\{x_i, x_{i+k}\}), w(\{x_{j-k}, x_j\})\}$ gives the maximum depth of any leaf node in $v$'s subtree.
    Together with the knowledge of $d(s,v)$, $v$ can compute $h(v)$.
\end{proof}
}

For the diameter, we use the following well-known lemma.
\lng{The proof is given for completeness.}

\begin{lemma}\label{lem:spa:tree:ecc_diameter}
    Let $G$ be a tree, $s \in V$ be an arbitrary node, and let $v \in V$ such that $d(s,v)$ is maximal.
    Then $\ecc(v) = D$.
\end{lemma}

\lng{
\begin{proof}
    Assume to the contrary that there is a node $u \in V$ such that $\ecc(u) > \ecc(v)$.
    Then there must be a node $w \in V$ such that $d(u,w) = \ecc(u) > \ecc(v) \ge d(v,w)$.
    Note that $d(u,w) > d(u,v)$, as otherwise $\ecc(u) \le \ecc(v)$, which would contradict our assumption.
    Let $P_1$ be the path from $s$ to $v$, $P_2$ be the path from $u$ to $w$, and let $t$ be the node in $P_2$ that is closest to $s$, i.e., $t = \argmin_{x \in P_2} d(s,x)$.

    If $t \notin P_1$, then let $x$ be the node farthest from $s$ that lies on $P_1$, and also on the path from $s$ to $t$ ($x$ might be $s$).
    Then $d(u,w) \le d(u,x) + d(x,w) \le d(v,x) + d(x,w) = d(v, w)$ (where $d(u,x) \le d(v,x)$ because $v$ is farthest to $s$), which contradicts $d(u,w) > d(v,w)$.
    
    If $t \in P_1$, $t$ must lie on a path from $v$ to $u$ or on a path from from $v$ to $w$.
    In the first case, $d(u,w) = d(u,t) + d(t,w) \le d(u,t) + d(t,v) = d(u,v)$, which implies $\ecc(v) \ge \ecc(u)$; the second case analogously implies $\ecc(v) \ge \ecc(w)$.
    Therefore, both cases lead to a contradiction with the assumption that $\ecc(v) < \ecc(u) = \ecc(w)$.
\end{proof}
}

Therefore, for the diameter it suffices to perform SSSP once from the node $s$ with highest identifier, then choose a node $v$ with maximum distance to $s$, and perform SSSP from $v$.
Since $\ecc(v) = D$, the node with maximum distance to $v$ yields the diameter.
Together with Lemma~\ref{lem:spa:aggregation}, we conclude the following theorem.

\begin{theorem} \label{thm:spa:diameter:tree}
    The diameter can be computed in any tree in time $O(\log n)$.
\end{theorem}

\section{Pseudotrees} \label{sec:spa:pseudotree}

Recall that a pseudotree is a graph that contains at most one cycle.
We define a \emph{cycle node} to be a node that is part of a cycle, and all other nodes as \emph{tree nodes}.
For each cycle node $v$, we define $v$'s tree $T_v$ as the connected component that contains $v$ in the graph in which $v$'s two adjacent cycle nodes are removed, and denote $h(v)$ as the height of $v$ in $T_v$. \shrt{Due to space constraints, we omit the details of the algorithm for pseudotrees and only give a brief description.
To compute SSSP, we first need to distinguish the cycle nodes from the tree nodes.
We do this by establishing rings of virtual nodes using the approach of Section~\ref{sec:spa:tree} (which must create two rings in a pseudotree).
Then, we can reduce the problem to computing SSSP in cycles and trees, for which we use the algorithms from the previous sections.}\lng{Before we show how SSSP and the diameter can be computed, we describe how the cycle can be identified, if it exists.

\begin{figure}[t]
    \centering
    \includegraphics[scale=0.75]{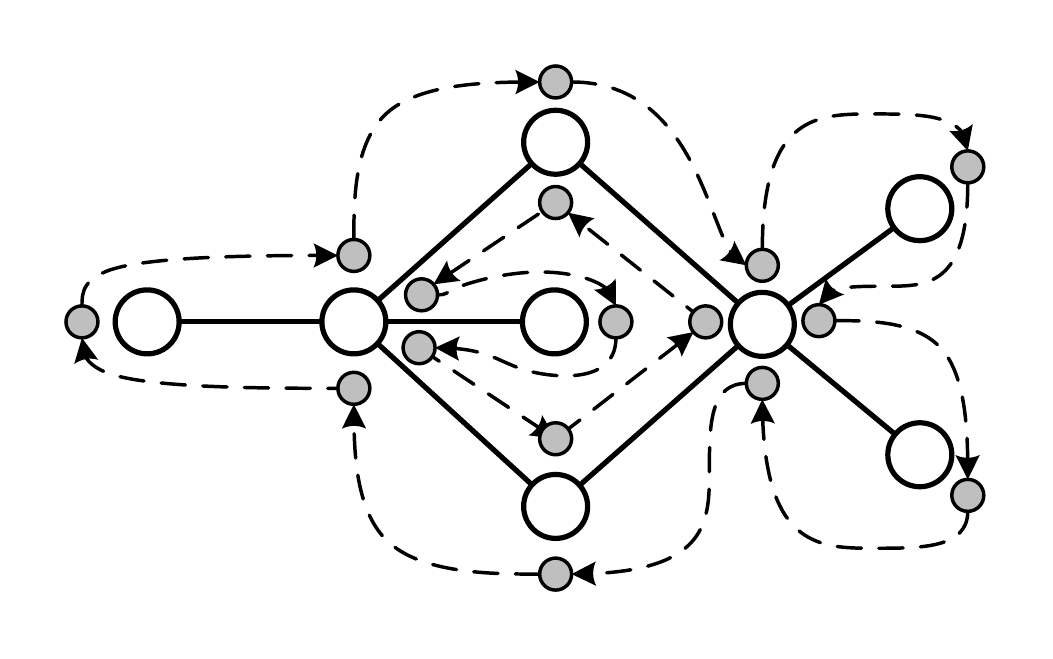}
\caption{Example for a pseudotree where each node $v$ emulates $\deg(v)$ virtual nodes (marked in grey) that form exactly two cycles (indicated by the dashed edges).}
\label{fig:virtualcycle_pseudotree}
\end{figure}

For this, we use the same approach as for the construction of the path $L$ in the tree.
We let each node $v$ simulate $\deg(v)$ virtual nodes $v_0,\ldots,v_{\deg(v)-1}$ and connect the virtual nodes according to the same rules as described in Section~\ref{sec:spa:tree}, with the exception that we do not leave out the last edge ending at $s$.
If there is no cycle, then this yields a single ring of virtual nodes, in which case we can use our previous algorithms.
Otherwise, this will create two rings of virtual nodes with the property that every cycle node must have at least one of its virtual nodes in each virtual ring\lng{ (see Figure~\ref{fig:virtualcycle_pseudotree} for an example)}.
Note that since nodes may have a high degree, we also need to redistribute the virtual nodes using the redistribution framework described in Section~\ref{sec:spa:tree}.
Since the arboricity of a pseudotree is at most $2$, we can compute an orientation with outdegree $6$ \cite[Theorem 3.5]{BE10}, and thus after redistributing the virtual nodes every node simulates at most $12$ virtual nodes.

To differentiate the at most two rings of virtual nodes from each other, we first establish shortcuts by performing the Introduction Algorithm on the virtual nodes.
Afterwards, every virtual node broadcasts its virtual identifier along all of its shortcuts; by repeatedly letting each node broadcast the highest identifier received so far for $O(\log n)$ rounds, each virtual node learns the maximum of all identifiers in its ring.
Any node whose virtual nodes learned different maxima must be a cycle node; if there exists no such node, which can easily be determined using Lemma~\ref{lem:spa:aggregation} in $G$, there is no cycle in $G$.
We conclude the lemma below.


\begin{lemma}\label{lem:spa:check_cycle_node}
    After $O(\log n)$ rounds every node $v \in V$ knows whether there is a cycle, and, if so, whether it is a cycle node.
\end{lemma}

\begin{proof}
    We argue the correctness of our construction by showing that if $G$ contains one cycle, then (1) the virtual nodes of each tree node are contained in the same virtual ring, (2) each cycle node has two virtual nodes contained in different virtual rings.
    For (1), let $v$ be a cycle node and $\{v,w\}$ be an edge to some tree node $w$.
    By our construction, there is exactly one virtual node $v_i$ of $v$ that is connected to a virtual node of $w$ and there is exactly one virtual node $w_i$ of $w$ that is connected to a virtual node $v_j$ of $v$.
    As presented in Section~\ref{sec:spa:tree}, this yields a path of virtual nodes starting at $v_i$, that traverses the subtree with root $w$ in a depth-first-search manner and ends at $v_j$, which implies (2).
    
    Specifically, this shows that the tree nodes do not introduce additional rings to our construction; therefore, we can disregard them and assume that $G$ forms a single cycle that does not contain any tree nodes.
    For this cycle it has to hold by our construction that every cycle node $v$ has exactly two virtual nodes $v_0$ and $v_1$ that are not directly connected to each other.
    This immediately implies that the virtual nodes have to form exactly two distinct rings of virtual nodes, since in case they would form one or more than two rings, there has to exist a ring node whose virtual nodes are connected to each other.
\end{proof}

Since we already know how to compute SSSP and the diameter on trees, for the remainder of this section we assume that $G$ contains a cycle.
In order to solve SSSP, we first perform our SSSP algorithm for tree graphs from source $s$ in the tree $T_v$ in which $s$ lies (note that $s$ may be $v$ itself).
Thereby, every node in $T_v$ learns its distance to $s$.
Specifically, $v$ learns $d(s,v)$, and can make this value known to all nodes by using Lemma~\ref{lem:spa:aggregation}.
After performing SSSP with source $v$ on the cycle nodes only, every cycle node $u \neq v$ knows $d(s, v) + d(v,u) = d(s,u)$, and can inform all nodes in its tree $T_u$ about $d(s,u)$ using Lemma~\ref{lem:spa:aggregation}.
Finally, $u$ performs SSSP in $T_u$ with source $u$, whereby each node $w \in T_u$ learns $d(s,u) + d(u, w) = d(s,w)$.
Together with Theorems~\ref{thm:spa:sssp:tree}, we obtain the following theorem.
} 

\begin{theorem}
    SSSP can be computed in any pseudotree in time $O(\log n)$.
\end{theorem}

\lng{ 
We now describe how to compute the diameter in a pseudotree.
In our algorithm, every cycle node $v$ contributes up to two \emph{candidates} for the diameter.
The first candidate for a node $v$ is the diameter of its tree $D(T_v)$.
If $\ecc(v) > h(v)$, then $v$ also contributes the value $\ecc(v) + h(v)$ as a candidate.
We first show how the values can be computed, and then prove that the maximum of all candidates, which can easily be determined using Lemma~\ref{lem:spa:aggregation}, is the diameter of $G$.

After $v$ has identified itself as a cycle node, it can easily compute its height $h(v)$ in time $O(\log n)$ by performing SSSP on $T_v$ from $v$ using Theorem~\ref{thm:spa:sssp:tree}, and then computing the maximum distance $d(v,u)$ of any node $u$ in $T_v$ using Lemma~\ref{lem:spa:aggregation}.
Furthermore, $D(T_v)$ can be computed in time $O(\log n)$ via an application of Theorem~\ref{thm:spa:diameter:tree}.

It remains to show how $v$ can learn $\ecc(v)$.
We define $m_\ell(v) := \max_{u \in V} h(u) - d_\ell(v,u)$, and $m_r(v) := \max_{u \in V} h(u) - d_r(v,u)$ (recall that $d_\ell(v,u)$ and $d_r(v,u)$ denote the distances from $v$ to $u$ along a left or right traversal of the cycle, respectively).

\begin{lemma} \label{lem:spa:pseudotree:ecc}
    Let $v \in V$ be a cycle node and let $v_\ell$ and $v_r$ be the left and right farthest nodes of $v$, respectively.
    $\ecc(v) = \max\{d_\ell(v,v_\ell) + m_r(v_\ell), d_r(v,v_r) + m_\ell(v_r)\}$.
\end{lemma}

\begin{proof}
    Let $t \in V$ such that $d(v,t) = \ecc(v)$, and let $u$ be a cycle node such that $t$ is a node of $T_u$.
    W.l.o.g., assume that $u$ lies on the \emph{right} side of $v$, i.e., $d_r(v,u) \le d_\ell(v,u)$.
    We define $d_\ell$ and $d_r$ to be $d_\ell(v,v_\ell)$ and $d_r(v,v_r)$, respectively.
    We show that (1) $d_r + m_\ell(v_r) \ge \ecc(v)$, and that (2) $d_\ell + m_r(v_\ell) \le \ecc(v)$ and $d_r + m_\ell(v_r) \le \ecc(v)$.
    Both statements together immediately imply the claim.

    For (1), note that $v_r$ will consider $u$ as a cycle node for the computation of $m_\ell(v_r)$, and thus $m_\ell(v_r) \ge h(u) - d_\ell(v_r,u)$.
    Therefore, we have that
    \[
        d_r + m_\ell(v_r) \ge d_r - d_\ell(v_r,u) + h(u) = d_r(v,u) + h(u) = d(v,t).
    \]

    For (2), we only show that $d_\ell + m_r(v_\ell) \le \ecc(v)$; the other side is analogous.
    Let $w$ be the node such that $m_r(v_\ell) = h(w) - d_r(v_\ell, w)$.
    First, assume that $w$ lies on the left side of $v$, i.e., $d_\ell(v,w) \le d_r(v,w)$.
    In this case, we have that $d_r(v_\ell,w) = d_\ell - d_\ell(v,w)$, which implies
    \begin{align*}
        m_r(v_\ell) &= h(w) - d_r(v_\ell, w) \\
        &= h(w) + d_\ell(v,w) - d_\ell \\
        &\le \ecc(v) - d_\ell.
    \end{align*}
    Now, assume that $w$ lies on the right side of $w$, in which case $d_r(v_\ell, w) = d_\ell + d_r(v,w)$.
    We have that
    \begin{align*}
        m_r(v_\ell) &= h(w) - d_r(v_\ell, w) \\
        &= h(w) - d_r(v,w) - d_\ell \\
        &\le h(w) + d_r(v,w) - d_\ell \\
        &\le \ecc(v) - d_\ell,
    \end{align*}
    which concludes the proof.
\end{proof}

Once each cycle node $v$ knows $m_\ell(v)$ and $m_r(v)$, every cycle node $u$ could easily infer its eccentricity by performing the diameter algorithm for the cycle of Theorem~\ref{thm:spa:sssp:cycle} to learn its farthest nodes.
The corresponding $m_\ell$ and $m_r$ values can be obtained alongside this execution.
Therefore, it remains to show how $v$ can compute $m_\ell(v)$ and $m_r(v)$.

To do so, the nodes first establish shortcuts along a left and right traversal of the cycle using the Introduction algorithm.\footnote{This time, \emph{each} node $v$ participates, and the initial left (right) neighbor of $v$ is $v$'s successor (predecessor) along a left traversal.}
Afterwards, every cycle node $v$ computes $m_\ell(v)$ (and, analogously, $m_r(v)$) in the following way.
$v$ maintains a value $x_v$, which will obtain the value $m_\ell(v)$ after $O(\log n)$ rounds.
Initially, $x_v := h(v)$.
In the first round, every cycle node $v$ sends $x_v - w(\{v, r_1\})$ to its right neighbor $r_1$.
When $v$ receives a value $x$ at the beginning of round $i$, it sets $x_v := \max\{x_v, x\}$ and sends $x_v - w(\{v,r_i\})$ to $r_{i}$.

\begin{lemma}\label{lem:spa:pseudotree:convergence}
    At the end of round $\lceil \log n \rceil + 1$, $x_v = m_\ell(v)$.
\end{lemma}

\begin{proof}
    We show that at the end of round $i \ge 1$, $$x_v = \max_{u \in V_\ell(v, i)} (h(u) - d_\ell(v,u)),$$ where $V_\ell(v, i)$ contains node $u \in V$ if the (directed) path from $v$ to $u$ in $G_\ell$ contains at most $2^{i-1}-1$ hops.
    The lemma follows from the fact that $V_\ell(v,\lceil \log n\rceil + 1) = V$.

    At the end of round $1$, $x_v = h(v)$, which establishes the inductive base since $v$ is the only node within $0$ hops from $v$.
    By the induction hypothesis, at the beginning of round $i > 1$ we have that $x_v = \max_{u \in V_\ell(v, i-1)} (h(u) - d_\ell(v,u))$.
    Furthermore, $v$ receives
    \begin{align*}
        x &= \max_{u \in V_\ell(\ell_{i-1},i-1)} (h(u) - d_\ell(\ell_{i-1},u)) - w(\{\ell_{i-1}, v\})  \\
        &= \max_{u \in V_\ell(\ell_{i-1},i-1)} (h(u) - d_\ell(v,u)).
    \end{align*}
    Since $V_\ell(v,i-1) \, \cup \, V_\ell(\ell_{i-1}, i-1) = V_\ell(v,i)$, we conclude that $\max\{x_v,x\} = \max_{u \in V_\ell(v, i)} (h(u) - d_\ell(v,u))$.
\end{proof}

Using the previous results, the nodes can now compute their candidates and determine the maximum of all candidates.
It remains to show the following lemma, from which we obtain Theorem~\ref{thm:spa:diameter_pseudotree}.

\begin{lemma}\label{lem:spa:pseudotree:candidate_agg}
    Let $C$ be the set of all candidates.
    $\max_{c \in C}\{c\} = D$.
\end{lemma}

\begin{proof}
    First, note that since every candidate value corresponds to the length of a shortest path in $G$, $c \le D$ for all $c \in C$.
    Let $s,t \in V$ be two nodes such that $D=d(s,t)$, and let $T_v$ and $T_w$ with cycle nodes $v$ and $w$ be the trees of $s$ and $t$, respectively.
    We show that $v$ or $w$ compute $D$ as one of their candidates.
    First, note that if one of the two nodes $s$ and $t$, say $s$, is a cycle node, then $D = \ecc(v) = \ecc(v) + h(v)$, and $\ecc(v) > h(v) = 0$; therefore, $v$ chooses $D$ as a candidate.

    Therefore, assume that both $s$ and $t$ are tree nodes.
    If $s$ and $t$ belong to the same tree, i.e., $v=w$, we have that $d(s,t)=D(T_v)$, which is a candidate of $v$.
    Otherwise, $D = \ecc(v)+h(v) = \ecc(w) + h(w)$.
    We only have to show that $\ecc(v) > h(v)$ or $\ecc(w) > h(w)$.
    Assume to the contrary that $\ecc(v) = h(v)$ and $\ecc(w) = h(w)$ (note that $\ecc(u) \ge h(u)$ for every cycle node $u$).
    Therefore, $\ecc(v) = d(v,w) + h(w) = h(v)$, and $\ecc(w) = d(v,w) + h(v) = h(w)$, which implies that $d(v,w) = 0$.
    However, this contradicts the assumption that $v \neq w$.
\end{proof}
} 
\shrt{
Computing the diameter is more complicated.
Since the longest path may not use any cycle node at all, each cycle node $v$ first contributes the diameter of its tree $T_v$ as a possible candidate. 
Furthermore, $v$ needs to compute its eccentricity $\ecc(v)$, and, if its eccentricity is larger than the height $h(v)$ of its tree $T_v$, contribute $\ecc(v) + h(v)$.
To compute its eccentricity, every cycle node needs to compute the distance to its farthest nodes using the algorithm of Theorem~\ref{thm:spa:sssp:cycle}, but also take into account the heights of the trees on the path to these nodes (as a longer path may lead into those trees).
}

\begin{theorem} \label{thm:spa:diameter_pseudotree}
    The diameter can be computed in any pseudotree in time $O(\log n)$.
\end{theorem}

\section{Cactus Graphs} \label{sec:spa:cactus_graph}

Our algorithm for cactus graphs relies on an algorithm to compute the maximal biconnected components (or \emph{blocks}) of $G$, where a graph is called biconnected if the removal of a single node would not disconnect the graph.
Note that for any graph, each edge lies in exactly one block.
In case of cactus graphs, each block is either a single edge or a simple cycle.
By computing the blocks of $G$, each node $v \in V$ classifies its incident edges into bridges (if there is no other edge incident to $v$ contained in the same block) and pairs of edges that lie in the same cycle.
To do so, we first give a variant of \cite[Theorem 1.3]{GHSW20} for the NCC$_0$ under the constraint that the input graph (which is not necessarily a cactus graph) has constant degree.
We point out how the lemma is helpful for cactus graphs, and then use a simulation of the biconnectivity algorithm of \cite{TV85} as in \cite[Theorem 1.4]{GHSW20} to compute the blocks of $G$.
The description and proofs of the following three lemmas are very technical and mainly describe adaptions of \cite{GHSW20}.
\shrt{Therefore, we defer them to the full version~\cite{full}.}

\begin{lemma}[Variant of {\cite[Theorem 1.3]{GHSW20}}] \label{lem:spa:constant_spanningtree}
    Let $G$ be any graph with constant degree.
    A spanning tree of $G$ can be computed in time $O(\log n)$, w.h.p., in the NCC$_0$.
\end{lemma}

\lng{
\begin{proof}
    To prove the lemma, we need a combination of \cite[Theorem 1.1]{GHSW20}, which transforms the initial graph into an overlay $G_T$ of diameter and degree $O(\log n)$, w.h.p., and a variant of the spanning tree algorithm of \cite[Theorem 1.3]{GHSW20}.
    
    \cite[Theorem 1.1]{GHSW20} creates a \emph{well-formed tree}, which is a tree that contains all nodes and has constant degree and diameter $O(\log n)$. 
    The tree is obtained from an intermediate graph $G_L$ that has degree and diameter $O(\log n)$. 
    The edges of $G_L$ are created by performing random walks of constant length in a graph $G_{L-1}$, which again is obtained from random walks in $G_{L-2}$, where $G_0$ is the graph $G$ extended by some self-loops and edge copies.
    More precisely, each edge of $G_i$ results from performing a random walk of constant length in $G_{i-1}$, and connecting the two endpoints of the walk.
    
    In the algorithm of \cite[Theorem 1.3]{GHSW20}, we first obtain a path $P_L$ in $G_L$ that contains all nodes using the Euler tour technique in a BFS tree of $G_L$.
    In contrast to the assumptions of the theorem, we do this on a graph $G_L$ that has degree $O(\log n)$, w.h.p. (c.f. \cite[Lemma 3.1 (1)]{GHSW20}, which implies that all graphs have degree $\Delta = O(\log n)$).
    This allows us to construct $P_L$ in time $O(\log n)$ in the NCC$_0$.
    We then follow the idea of \cite{GHSW20} to iteratively replace each edge of $P_{i}$ by the edges the random walk that resulted in that edge took in $G_{i-1}$ to obtain a path $P_{i-1}$ in $G_{i-1}$ for all $1 \le i \le L$.
    Since each random walk only took a constant number of steps, the endpoints of each edge in $G_i$ can easily inform the corresponding edge's endpoints in $G_{i-1}$.
    After $L=O(\log n)$ steps, $P_0$ only contains edges of $G$ (or self-loops, which can easily be removed).
    
    In \cite{GHSW20}, every node then computes the first edge over which it is reached in $P_0$ and keeps only this edge, which results in a spanning tree.
    However, note that in our case each graph $G_i$ has a degree of $O(\log n)$, w.h.p., which implies that a node may have a degree of $O(\log^2 n)$ in $P_0$.
    This prevents us from computing the first edge of each node by performing pointer jumping in $P_0$ in $O(\log n)$ rounds in the NCC$_0$.
    In the algorithm of \cite{GHSW20}, where the degree is even higher, the authors simply allow higher communication work.
    Instead, we compute the first edge of each node "on the fly" while constructing $P_{i-1}$ from $P_i$.
    To do so, we maintain the invariant that the edges of $P_i$ are enumerated with unique $O(\log n)$ bit labels such that for every two edges $e, e'$ in $P$ the label of $e$ is smaller than the label of $e'$ if and only if $e$ occurs before $e'$ in $P_i$, and mark the endpoint of each edge that occurs earlier in $P_i$.
    At the beginning, for $P_L$, we simply enumerate the edges of $P_L$ from $1$ to $k$ starting at the root of the BFS tree using pointer jumping.
    This can be done since the degree of each $G_i$ is $O(\log n)$, w.h.p.
    Then, whenever an edge $e$ of $P_i$ is replaced by edges of $G_{i-1}$, we enumerate the corresponding edges using the label of $e$ extended by a constant number of bits to order them along the random walk from the earlier endpoint of $e$ to its other endpoint.
    Then, the edges of $P_0$ will be ordered, and each node can locally determine its first edge without the need to perform pointer jumping.
    This concludes the lemma.
\end{proof}

Now let $G$ be a cactus graph.
To use the previous lemma on cactus graphs, we follow the idea of \cite[Section 4.2]{GHSW20} and transform $G$ into a constant-degree graph $G'$ using a construction similar to \emph{child-sibling trees} \cite{AW07, GHSS17}, then compute a spanning tree $S'$ on $G'$ using the previous lemma, and finally infer a spanning tree $S$ of $G$ from $S'$.
The details can be found in the proof of the following lemma.
}

\begin{lemma} \label{lem:spa:cactus_spanning}
   A spanning tree of a cactus graph $G$ can be computed in time $O(\log n)$, w.h.p.
\end{lemma}

\lng{
\begin{proof}
    We first use \cite[Theorem 3.5]{BE10} to assign directions to the edges; since the arboricity of $G$ is 2, each node $v$ will have $6$ outgoing neighbors (i.e., the nodes $u$ such that $\{v,u\}$ is directed towards $u$), and possibly many incoming neighbors.
    Let $v_1, \ldots, v_k$ be the incoming neighbors of $v$ sorted by increasing identifier.
    Let $G'$ be the graph that contains the edge $\{v, v_1\}$ for each $v$ and an edge $\{v_i, v_{i+1}\}$ for all $1 \le i < k$ (multi-edges created in this way are regarded as a single edge).
    Note that each node has at most 6 outgoing neighbors, keeps only one edge to an incoming neighbor, and is assigned at most one additional neighbor by each of its outgoing neighbors.
    Therefore, $G'$ has degree at most 13 and we can apply Lemma~\ref{lem:spa:constant_spanningtree} to obtain a spanning tree $S'$.
    However, $S'$ is actually a spanning tree of $G'$, which contains edges $\{v_i, v_{i+1}\}$ that do not exist in $G$.
    To obtain a spanning tree of $G$ instead, we follow the approach of \cite[Lemma 4.12]{GHSW20}:
    Recall that Lemma~\ref{lem:spa:constant_spanningtree} creates a path $P_0$ that contains all nodes of $G'$ and in which the edges are enumerated by their order in $P_0$.
    Every edge $\{v_i, v_{i+1}\}$ that appears in $P_0$, and that does not exist in $G$, can now easily be replaced by the two edges $\{v, v_i\}$ and $\{v, v_{i+1}\}$ using local communication (i.e., $v_i$ and $v_{i+1}$ were incoming neighbors of $v$).
    Furthermore, we can easily keep the order of the edges so that afterwards each node still knows the first edge over which it is reached in the path, whereby we obtain a spanning tree of $G$.
\end{proof}

To obtain the biconnected components of $G$, we now perform a simulation of \cite{TV85} in almost the same way as \cite[Theorem 1.4]{GHSW20}.
The algorithm relies on a spanning tree, which we compute using Lemma~\ref{lem:spa:cactus_spanning}, and constructs a helper graph, whose connected components yield the biconnected components of $G$.
The only difference from our application to the simulation described in \cite{GHSW20} lies in the fact that we do not rely on the complicated algorithm of  \cite[Theorem 1.2]{GHSW20} for computing the connected components, but use Lemma~\ref{lem:spa:constant_spanningtree}, and the child-sibling approach of Lemma~\ref{lem:spa:cactus_spanning}, exploiting the properties of a cactus graph.
}

\begin{lemma} \label{lem:spa:biconnected}
   The biconnected components of a cactus graph $G$ can be computed in time $O(\log n)$, w.h.p.
\end{lemma}

\lng{
\begin{proof}
    The only difference to the approach of \cite[Theorem 1.4]{GHSW20} lies in Step 4. 
    Instead of using \cite[Theorem 1.2]{GHSW20} to compute the connected components of the helper graph $G''$, we observe that every node $v$ is only adjacent to at most one node $w$ in $G$ that lies in a different subtree than $v$ in the rooted spanning tree $T$ of $G$.
    Therefore, its parent edge $\{v,u\}$ in $T$ (which is a node in $G''$) will only create one connection to the parent edge $\{w, x\}$ of $w$ according to Rule 1 of Step 3 of \cite[Section 4.4]{GHSW20}.
    Therefore, the arboricity of $G''$ is constant, and we can combine the child-sibling approach of the proof of Lemma~\ref{lem:spa:cactus_spanning} with Lemma~\ref{lem:spa:constant_spanningtree} to compute a spanning tree on each connected component of $G''$.
    Note that, as \cite{GHSW20} point out, this simulation is possible since local communication in $G''$ can be carried out using a constant number of local communication rounds in $G$, and regarding the global communication we observe that each node in $G$ only simulates a single node in $G''$. 
    Finally, using a simulation of Lemma~\ref{lem:spa:aggregation} on the spanning trees, we can distinguish the connected components from one another, which concludes the lemma.
\end{proof}
}

Thus, every node can determine which of its incident edges lie in the same block in time $O(\log n)$, w.h.p.
Let $s$ be the source for the SSSP problem.
First, we compute the \emph{anchor node} of each cycle in $G$, which is the node of the cycle that is closest to $s$ (if $s$ is a cycle node, then the anchor node of that cycle is $s$ itself).
To do so, we replace each cycle $C$ in $G$ by a binary tree $T_C$ of height $O(\log n)$ as described in \cite{GHSS17}.
More precisely, we first establish shortcut edges using the Introduction algorithm in each cycle, and then perform a broadcast from the node with highest identifier in $C$ for $O(\log n)$ rounds.
If in some round a node receives the broadcast for the first time from $\ell_{i}$ or $r_i$, it sets that node as its parent in $T_C$ and forwards the broadcast to $\ell_j$ and $r_j$, where $j = \min\{i-1, 0\}$.
After $O(\log n)$ rounds, $T_C$ is a binary tree that contains all nodes of $C$ and has height $O(\log n)$.
To perform the execution in all cycles in parallel, each node simulates one virtual node for each cycle it lies in and connects the virtual nodes using their knowledge of the blocks of $G$.
To keep the global communication low, we again use the redistribution framework described in Section~\ref{sec:spa:tree} (note that the arboricity of $G$ is $2$).

\begin{lemma} \label{lem:spa:cactus_anchor}
    Let $T$ be the (unweighted) tree that results from taking the union of all trees $T_C$ and all bridges in $G$.
    For each cycle $C$, the node $a_C := \argmin_{v \in C} d_T(s,v)$ is the anchor node of $C$.
\end{lemma}

The correctness of the lemma above simply follows from the fact that any shortest path from $s$ to any node in $C$ must contain the anchor node of $C$ both in $G$ and in $T$.
Therefore, the anchor node of each cycle can be computed by first performing the SSSP algorithm for trees with source $s$ in $T$ and then conducting a broadcast in each cycle.
Now let $v$ be the anchor node of some cycle $C$ in $G$.
By performing the diameter algorithm of Theorem~\ref{thm:spa:sssp:cycle} in $C$, $v$ can compute its left and right farthest nodes $v_\ell$ and $v_r$ in $C$.
Again, to perform all executions in parallel, we use our redistribution framework.

\begin{lemma} \label{lem:spa:cactus_spt}
    Let $S_G$ be the graph that results from removing the edge $\{v_\ell, v_r\}$ from each cycle $C$ with anchor node $v$.
    $S_G$ is a shortest path tree of $G$ with source $s$.
\end{lemma}

\lng{
\begin{proof}
    Since we delete one edge of each cycle, and every edge is only contained in exactly one cycle, $S_G$ is a tree.
    Assume to the contrary that $S_G$ does not contain a shortest path from $s$ to any other node, i.e., there exists a $P$ path from $s$ to some node $t$ in $G$ that is shorter than the (unique) shortest path from $s$ to $t$ in $S_G$.
    Specifically, $P$ must contain an edge $e$ of a cycle $C$ such that the subpath $P' = (v, \ldots, u)$ of $P$ that contains only the nodes of $C$ is strictly shorter than the path from $v$ to $u$ in $S_G$ (roughly speaking, not all subpaths over cycles that use the deleted edge of that cycle can have the same length as the subpath that goes along the other side of the cycle).
    Note that $v$ is the anchor node of $C$, and $P'$ is the (unique) path from $v$ to $u$ in $C$ that contains $e$, whereas $S_G$ contains the unique path $P''$ from $v$ to $u$ that does \emph{not} contain $e$ (i.e., it goes along the other direction of the cycle).
    However, by definition of $e$, $P''$ must already be a shortest path between $v$ and $u$ in $G$, which contradicts the assumption that $P'$ is shorter.
\end{proof}
}

Therefore, we can perform the SSSP algorithm for trees of Theorem~\ref{thm:spa:sssp:tree} on $S_G$ and obtain the following theorem.

\begin{theorem} \label{thm:spa:sssp:cactus}
    SSSP can be computed in any cactus graph in time $O(\log n)$.
\end{theorem}

To compute the diameter, we first perform the algorithm of Lemma~\ref{lem:spa:cactus_spt} with the node that has highest identifier as source $s$,\footnote{In the NCC$_0$, this node can be determined by constructing the tree $T$ from Lemma~\ref{lem:spa:cactus_anchor} and using Lemma~\ref{lem:spa:aggregation} on $T$.} which yields a shortest path tree $S_G$.
This tree can easily be rooted using Lemma~\ref{lem:spa:tree_rooting}.
Let $Q(v)$ denote the children of $v$ in $S_G$.
Using Lemma~\ref{lem:spa:height}, each node $v$ can compute its height $h(v)$ in $S_G$ and can locally determine the value \lng{\[
    m(v) := \max_{u,w \in Q(v), u \neq w} ( h(u) + h(w) + w({v, u}) + w({v, w}) ).
\]}\shrt{$m(v) := \max_{u,w \in Q(v), u \neq w} ( h(u) + h(w) + w({v, u}) + w({v, w}) )$.}
$m(v)$ is the length of the longest path in $v$'s subtree in $S_G$ that contains $v$.
We further define the pseudotree $\Pi_C$ of each cycle $C$ as the graph that contains all edges of $C$ and, additionally, an edge $\{v, t_v\}$ for each node $v \neq a_C$ of $C$, where $t_v$ is a node that is simulated by $v$, and $w(\{v, t_v\}) = \max_{u \in Q(v) \setminus C} (h(u) + w(\{v, u\}))$.
Intuitively, each node $v$ of $C$ that is not the anchor node is attached an edge whose weight equals the height of its subtree in $S_G$ without considering the child of $v$ that also lies in $C$ (if that exists).
Then, for each cycle $C$ in parallel, we perform the algorithm of Theorem~\ref{thm:spa:diameter_pseudotree} on $\Pi_C$ to compute its diameter $D(\Pi_C)$ (using the redistribution framework).
We obtain the diameter of $G$ as the value
\lng{\[
    \hat{D} := \max(\max_{v \in V}(m(v)), \max_{\text{cycle }C}(D(\Pi_C))).
\]}\shrt{$\hat{D} := \max(\max_{v \in V}(m(v)), \max_{\text{cycle }C}(D(\Pi_C)))$.}
By showing that $\hat{D} = D$, we conclude the following theorem.

\begin{theorem}
    The diameter can be computed in any cactus graph in time $O(\log n)$.
\end{theorem}

\lng{
\begin{proof}
    We first show that $\hat{D} \le D$, and then prove $D \le \hat{D}$.
    For the first part, note that $m(v) \le D$ for all $v$, since $S_G$ is a shortest path tree, and thus $m(v)$ corresponds to the length of some shortest path.
    Furthermore, the weight of each attached edge $\{v, t_v\}$ of a cycle $C$ corresponds to the length of the path from $v$ to some descendant of $v$ in $S_G$, which must be a shortest path.
    Therefore, for any shortest path in $\Pi_C$, a shortest path of the same length must exist in $G$, which implies that $D(\Pi_C) \le D$.
    We conclude that $\hat{D} \le D$.
    
    For the second part, let $P = (v_1, \ldots, v_k)$ be a longest shortest path in $G$.
    First, note that each cycle in $G$ is only entered and left at most once by $P$ (if it is left at some node, it may only be entered again at the same node, which is impossible since we have positive edge weights).
    We first slightly change $P$ to ensure that it does not simultaneously contain the deleted edge and the anchor node of the same cycle.
    Consider any (maximal) subpath $P_C$ of $P$ that contains edges of a cycle $C$ in $G$, and assume that $P'$ contains the deleted edge $e_C$ of $C$ (i.e., the edge that is incident to the farthest nodes of $C$'s anchor node $a_C$ in the cycle).
    If $P_C$ contains both $e_C$ and $a_C$, then the cycle must be \emph{symmetric}: $P_C$ begins at $e_C$ and ends at $a_C$ (or vice versa), and the other side of the cycle has the same weight as $P_C$.
    Therefore, we can simply replace $P_C$ by the other edges of $C$.
    After replacing each subpath $P_C$ of $P$ for each cycle $C$, either $P$ (1) only contains edges of $S_G$, (2) contains a deleted edge $e_C$.
    In the second case, $P$ cannot contain $a_C$, and therefore every other cycle contained in $P$ must be entered over its anchor node, which implies that no other deleted edge can be contained in $P$.
    
    Assume that case (1) holds.
    In this case, $P$ can be divided into two paths $P_1 = (v_1, \ldots, v_j)$, where $\{v_i, v_{i+1}\}$ is directed from $v_i$ to $v_{i+1}$ for all $1 \le i < j$, and $P_2 = (v_{j}, \ldots, v_k)$, where $\{v_i, v_{i+1}\}$ is directed from $v_{i+1}$ to $v_{i}$ for all $j \le i < k$ (note that $P_1$ and $P_2$ may also be empty).
    Therefore, $D = w(P) = w(P_1) + w(P_2) \le m_{v_j} \le \hat{D} $.
    
    Finally, assume that case (2) holds.
    Let $e_C$ be the single deleted edge of $P$ and $P_C$ be the subpath of $P$ contained in $C$.
    By the arguments above, $P_C$ does not contain $a_C$.
    Let $u$ and $v$ be the two endpoints of $P_C$, where $u$ is visited in $P$ before $v$, and let $P = P_u \circ P_C \circ P_v$, where $P_u$ is the subpath from $P$'s first node to $u$, and $P_v$ is the subpath from $v$ to the last node to $P$.
    Since $P_u$ and $P_v$ do not contain any deleted edge, nor any edge of $C$, it holds that $w(P_u) \le w(\{u, t_u\})$ (recall that the weight of $u$'s virtual edge is the height of $u$ in $S_G$ without the subtree of $u$'s child in $C$), and $w(P_v) \le w(\{v, t_v\})$.
    Let $P'$ be the path in $\Pi_C$ that starts at $\{u, t_u\}$, then follows $P_C$, and ends at $\{v, t_v\}$.
    Since $P_C$ is a shortest path in $C$, $P'$ is a shortest path in $\Pi_C$ and $w(P) \le w(P')$, which implies that $D = w(P) \le w(P') \le D(\Pi_C) \le \hat{D}$.
\end{proof}
}

\section{Sparse Graphs} \label{sec:spa:sparse_graph}

In this final section, we present constant factor approximations for SSSP and the diameter in graphs that contain at most $n + O(n^{1/3})$ edges and that have arboricity at most $O(\log n)$.
Our algorithm for such graphs relies on an MST $M = (V, E')$ of $G$, where $E' \subseteq E$.
$M$ can be computed deterministically in time $O(\log^2 n)$ using \cite{GHSS17}, Observation 4, in a modified way\footnote{The algorithm of \cite{GHSS17} computes a (not necessarily minimum) spanning tree, which would actually already suffice for the results of this paper.
However, if $G$ contains edges with exceptionally large weights, an MST may yield much better results in practice.}.

\begin{lemma} \label{lem:spa:mst}
    The algorithm computes an MST of $G$ deterministically in time $O(\log^2 n)$.
\end{lemma}

\lng{
\begin{proof}
    The \emph{Overlay Construction Algorithm} presented in \cite{GHSS17} constructs a low-diameter overlay in time $O(\log n)$ by alternatingly grouping and merging supernodes until a single supernode remains.
    As a byproduct, Observation 4 remarks that the edges over which merge requests have been sent from one supernode to another form a spanning tree.

    To obtain an MST, we change the way a supernode $u$ chooses a neighboring supernode to merge with.
    More specifically, as many other distributed algorithms for MST computation, our modification directly mimics the classic approach of Borůvka \cite{NMN01}.
    Instead of choosing the adjacent supernode $v$ that has the highest identifier, and sending a merge request if $v$'s identifier is higher than $u$'s identifier, $u$ determines the outgoing edge (i.e., the edge incident to a node of $u$ whose other endpoint is not in $u$) with smallest weight, breaking ties by choosing the edge with smallest identifier (where the identifier of an edge $\{x,y\}$, $\id(x) < \id(y)$, is given by $\id(x) \circ \id(y)$).
    It is well-known that the edges chosen in this way form an MST.

    Compared to the grouping stage described in \cite{GHSS17}, this yields components of supernodes that form pseudotrees with a cycle of length $2$ (see, e.g., \cite{JM95}).
    However, such cycles can easily be resolved locally by the supernodes such that the resulting components form trees, which allows us to perform the merging stage of \cite{GHSS17} without any further modifications.
\end{proof}
}

We call each edge $e \in E \setminus E'$ a \emph{non-tree edge}.
Further, we call a node \emph{shortcut node} if it is adjacent to a non-tree edge, and define $\Sigma \subseteq V$ as the set of shortcut nodes.
Clearly, after computing $M$ every node $v \in \Sigma$ knows that it is a shortcut node, i.e., if one of its incident edges has not been added to $E'$.
In the remainder of this section, we will compute approximate distances by (1) computing the distance from each node to its closest shortcut node in $G$, and (2) determining the distance between any two shortcut nodes in $G$.
For any $s,t \in V$, we finally obtain a good approximation for $d(s,t)$ by considering the path in $M$ as well as a path that contains the closest shortcut nodes of both $s$ and $t$.

Our algorithms rely on a \emph{balanced decomposition tree} $T_M$, which allows us to quickly determine the distance between any two nodes in $G$, and which is presented in Section~\ref{sec:spa:sparse:tree_decomposition}.
In Section~\ref{sec:spa:sparse:nearest_cut_node}, $T_M$ is extended by a set of edges that allow us to solve (1) by performing a distributed multi-source Bellman-Ford algorithm for $O(\log n)$ rounds.
For (2), in Section~\ref{sec:spa:sparse:apsp_cut_nodes} we first compute the distance between any two shortcut nodes in $M$, and then perform matrix multiplications to obtain the pairwise distances between shortcut nodes in $G$.
By exploiting the fact that $|\Sigma| = O(n^{1/3})$, and using techniques of \cite{AGG+19}, we are able to distribute the $\Theta(n)$ operations of each of the $O(\log n)$ multiplications efficiently using the global network.
In Section~\ref{sec:spa:sparse:approx}, we finally show how the information can be used to compute $3$-approximations for SSSP and the diameter.

For simplicity, in the following sections we assume that $M$ has degree $3$.
Justifying this assumption, we remark that $M$ can easily be transformed into such a tree while preserving the distances in $M$.
First, we root the tree at the node with highest identifier using Lemma~\ref{lem:spa:tree_rooting}.
Then, every node $v$ replaces the edges to its children by a binary tree of virtual nodes, where the leaf nodes are the children of $v$, the edge from each leaf $u$ to its parent is assigned the weight $w(\{v,u\})$, and all inner edges have weight $0$.\footnote{Note that the edge weights are no longer strictly positive; however, one can easily verify that the algorithms of this section also work with non-negative edge weights.}
The virtual nodes are distributed evenly among the children of $v$ such that each child is only tasked with the simulation of at most one virtual node.
Note that the virtual edges can be established using the local network.

\subsection{Hierarchical Tree Decomposition} \label{sec:spa:sparse:tree_decomposition}
We next present an algorithm to compute a hierarchical tree decomposition of $M$, resulting in a \emph{balanced decomposition tree} $T_M$.
$T_M$ will enable us to compute distances between nodes in $M$ in time $O(\log n)$, despite the fact that the diameter of $M$ may be very high.

Our algorithm constructs $T_M$ as a binary rooted tree $T_M=(V,E_T)$ of height $O(\log n)$ with root $r \in V$ (which is the node that has highest identifier) by selecting a set of global edges $E_T$.
Each node $v \in V$ knows its parent $p_T(u) \in V$.
To each edge $\{u,v\} \in E_T$ we assign a weight $w(\{u,v\})$ that equals the sum of the weights of all edges on the (unique) path from $u$ to $v$ in $M$.
Further, each node $v \in V$ is assigned a distinct label $l(v) \in \{0,1\}^{O(\log n)}$ such that $l(v)$ is a prefix of $l(u)$ for all children $u$ of $v$ in $T_M$, and $l(r) = \varepsilon$ (the empty word).

From a high level, the algorithm works as follows.
Starting with $M$, within $O(\log n)$ iterations $M$ is divided into smaller and smaller components until each component consists of a single node.
More specifically, in iteration $i$, every remaining component $A$ handles one \emph{recursive call} of the algorithm, where each recursive call is performed independently from the recursive calls executed in other components.
The goal of $A$ is to select a \emph{split node} $x$, which becomes a node at depth $i-1$ in $T_M$, and whose removal from $M$ divides $A$ into components of size at most $|A|/2$.
The split node $x$ then recursively calls the algorithm in each resulting component; the split nodes that are selected in each component become children of $x$ in $T_M$\shrt{.}\lng{ (see Figure~\ref{fig:tree_decomposition}).}

When the algorithm is called at some node $v$, it is associated with a \emph{label} parameter $l \in \{0,1\}^{O(\log n)}$ and a \emph{parent} parameter $p \in V$.
The first recursive call is initiated at node $r$ with parameters $l = \varepsilon$ and $p = \emptyset$.
Assume that a recursive call is issued at $v \in V$, let $A$ be the component of $M$ in which $v$ lies, and let $A_1, A_2$ and $A_3$ be the at most three components of $A$ that result from removing $v$.
Using Lemma~\ref{lem:spa:subtree_size_aggregate}, every node $u$ in $A_1$ can easily compute the number of nodes that lie in each of its adjacent subtrees in $A_1$ (i.e., the size of the resulting components of $A_1$ after removing $u$).
It is easy to see that there must be a \emph{split node} $x_1$ in $A_1$ whose removal divides $A_1$ into components of size at most $|A|/2$ (see, e.g., \cite[Lemma 4.1]{AHK+20}); if there are multiple such nodes, let $x_1$ be the one that has highest identifier.
Correspondingly, there are split nodes $x_2$ in $A_2$ and $x_3$ in $A_3$.
$v$ learns $x_1, x_2$ and $x_3$ using Lemma~\ref{lem:spa:aggregation} and sets these nodes as its children in $T_M$.
By performing the SSSP algorithm of Theorem~\ref{thm:spa:sssp:tree} with source $v$ in $A_1$, $x_1$ learns $d_M(x_1,v)$, which becomes the weights of the edge $\{v, x_1\}$ (correspondingly, the edges $\{v, x_2\}$ and $\{v, x_3\}$ are established). 
To continue the recursion in $A_1$, $x$ calls $x_1$ with label parameter $l \circ 00$ and parent parameter $v$.
Correspondingly, $x_2$ is called with $l \circ 01$, and $x_3$ with $l \circ 10$.

\lng{
\begin{figure}[t]
	\centering
	\begin{subfigure}{180pt}
	\centering
	\includegraphics{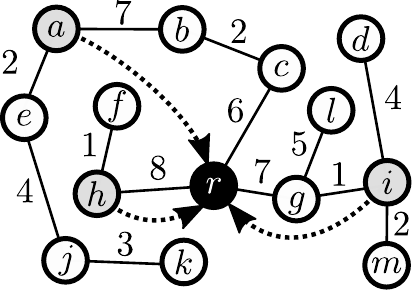}
	\subcaption{\centering}
	\label{fig:tree_decomposition_1}
    \end{subfigure}
    \begin{subfigure}{180pt}
	\centering
	\includegraphics{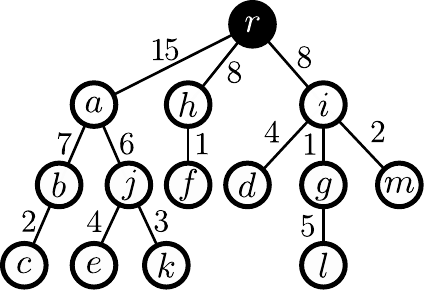}
	\subcaption{\centering}
	\label{fig:tree_decomposition_2}
    \end{subfigure}
    \caption{
    (a) The graph $M$ after the first step of the tree decomposition.
    The black node is the root $r$, and the grey nodes are the first split nodes chosen for each of $r$'s subtrees.
    The algorithm will recursively be called in each connected components of white nodes.
    (b) The resulting balanced decomposition tree $T_M$.
    }
    \label{fig:tree_decomposition}
\end{figure}
}

\begin{theorem} \label{thm:spa:decomposition_tree}
    A balanced decomposition tree $T_M$ for $M$ can be computed in time $O(\log^2 n)$.
\end{theorem}

\lng{
\begin{proof}
    It is easy to see that our algorithm constructs a correct balanced decomposition tree.



    It remains to analyse the runtime of our algorithm.
    In each recursive call we need $O(\log n)$ rounds to compute the sizes of all subtrees for any node (Lemma~\ref{lem:spa:subtree_size_aggregate}) and $O(\log n)$ rounds to find a split node (Lemma~\ref{lem:spa:aggregation}).
    Computing the weight of a global edge chosen to be in $E_T$ takes $O(\log n)$ rounds (Theorem~\ref{thm:spa:sssp:tree}).
    Since the component's sizes at least halve in every iteration, the algorithm terminates after $O(\log n)$ iterations.
    This proves the theorem.
\end{proof}
}

It is easy to see that one can route a message from any node $s$ to any node $t$ in $O(\log n)$ rounds by following the unique path in the tree from $s$ to $t$, using the node labels to find the next node on the path.
However, the sum of the edge's weights along that path may be higher than the actual distance between $s$ and $t$ in $M$.

\subsection{Finding Nearest Shortcut Nodes} \label{sec:spa:sparse:nearest_cut_node}
To efficiently compute the nearest shortcut node for each node $u \in V$, we extend $T_M$ to a \emph{distance graph} $D_T = (V, E_D)$, $E_D \supseteq E_T$, by establishing additional edges between the nodes of $T_M$.
Specifically, unlike $T_M$, the distance between any two nodes in $D_T$ will be equal to their distance in $M$, which allows us to employ a distributed Bellman-Ford approach. 

\lng{
\begin{figure}[t]
    \centering
    \includegraphics[scale=0.95]{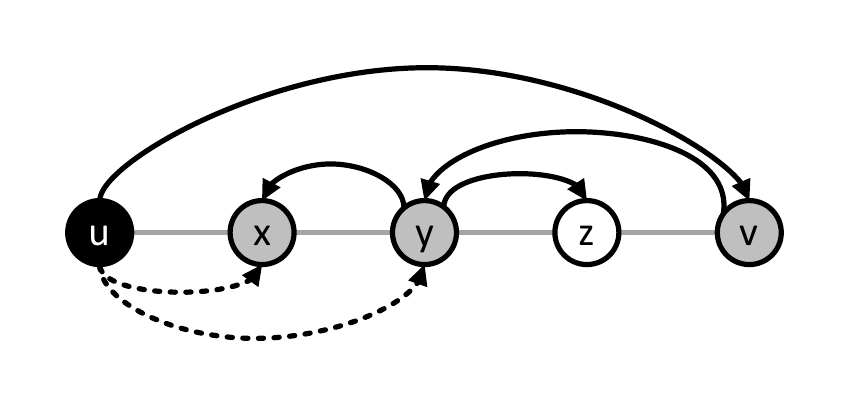}
\caption{Example for the construction of additional edges (indicated by the dashed lines) going into the node $u$. Grey edges are edges of $M$, straight black edges are edges in $T_M$. $u$ is $v$'s parent in $T_M$, with $x,y$ and $z$ being in the subtree of $v$ in $T_M$. We always choose the node in the subtree that goes in the direction back to $u$ in $M$, so we add edges $\{y,u\}$ and $\{x,u\}$. We do not add an edge $\{z,u\}$ because its subtree follows the opposite direction of $u$ from the perspective of $y$.
The descendants of $u$ in $D_T$ are marked grey.}
\label{fig:additional_edges_example}
\end{figure}
}
We describe the algorithm to construct $D_T$ from the perspective of a fixed node $u \in V$\lng{
(for an illustration, see Figure~\ref{fig:additional_edges_example})}.
For each edge $\{u,v\} \in E_T$ such that $u = p_T(v)$ for which there does \emph{not} exist a local edge $\{u,v\} \in E'$, we know that the edge $\{u,v\}$ "skips" the nodes on the unique path between $u$ and $v$ in $M$. 
Consequently, these nodes must lie in a subtree of $v$ in $T_M$.
Therefore, to compute the exact distance from $u$ to a skipped node $w$, we cannot just simply add up the edges in $E_T$ on the path from $u$ to $w$, as this sum must be larger than the distance $d(u,w)$.

To circumvent this problem, $u$'s goal is to establish additional edges to some of these skipped nodes.
Let $x \in V$ be the neighbor of $u$ in $M$ that lies on the unique path from $u$ to $v$ in $M$.
To initiate the construction of edges in each of its subtrees, $u$ needs to send messages to \emph{each} child $v$ in $T_M$ that skipped some nodes (recall that $u$ is able to do so because it has degree $3$ in $T_M$).
Such a message to $v$ contains $l(x)$, $l(u)$, $\id(u)$ and $w(\{u,v\})$.
Upon receiving the call from $u$, $v$ contacts its child node $y$ in $T_M$ whose label is a prefix of $l(x)$, forwarding $u$'s identifier, $l(x)$ and the (updated) weight $w(\{y,u\}) = w(\{u,v\})-w(\{v,y\})$.
$y$ then adds the edge $\{y,u\}$ with weight $w(\{y,u\})$ to the set $E_D$ by informing $u$ about it.
Then, $y$ continues the recursion at its child in $T_M$ that lies in $x$'s direction, until the process reaches $x$ itself.
Since the height of $T_M$ is $O(\log n)$, $u$ learns at most $O(\log n)$ additional edges and thus its degree in $D_T$ is $O(\log n)$.

Note that since the process from $u$ propagates down the tree level by level, we can perform the algorithm at all nodes in parallel, whereby the separate construction processes follow each other in a pipelined fashion without causing too much communication.
Together with Theorem~\ref{thm:spa:decomposition_tree}, we obtain the following lemma.

\begin{lemma} \label{lem:spa:distance_graph}
    The distance graph $D_T=(V,E_D)$ for $M$ can be computed in time $O(\log^2 n)$.
\end{lemma}

From the way we construct the node's additional edges in $E_D$, and the fact that the edges in $E_T$ preserve distances in $M$, we conclude the following lemma.

\begin{lemma}
    For any edge $\{u,v\} \in E_D$ it holds $w(\{u,v\}) = d_M(u,v)$, where $d_M(u,v)$ denotes the distance between $u$ and $v$ in $M$.
\end{lemma}

The next lemma is crucial for showing the correctness of the algorithms that follow.

\begin{lemma} \label{lem:spa:distance_graph_correctness}
For every $u, v \in V$ we have that (1) every path from $u$ to $v$ in $D_T$ has length at least $d_{M}(u,v)$, and (2) there exists a path $P$ with $w(P) = d_M(u,v)$ and $|P| = O(\log n)$ that only contains nodes of the unique path from $u$ to $v$ in $T_M$.
\end{lemma}

\lng{
\begin{proof}
    For (1), assume to the contrary that there is a path $P$ from $u$ to $v$ with length less than $d_M(u,v)$.
    As no path in $M$ from $u$ to $v$ can be shorter than $d_M(u,v)$, $P$ contains at least one edge from $E_D$.
    However, by the way we construct the weights of the edges in the distance graph $D_T$, it holds for any edge $\{x,y\} \in E_D$ that $w(\{x,y\}) = d_M(x,y)$.
    As $P$ can be any arbitrary path from $u$ to $v$, we get $w(P) \geq d_M(u,v)$, which is a contradiction.

    For (2), let $P_T = (u = x_0,x_1,x_2,\ldots,x_m = v)$ be the path from $u$ to $v$ consisting of edges $E_T$.
    By the construction of $T$, $|P_T| = O(\log n)$.
    In case $w(P_T) = d_M(u,v)$ we are done, so let us assume that $w(P_T) > d_M(u,v)$.
    We show that we can replace subpaths of $P_T$ by single edges out of $E_D \setminus E_T$ until we arrive at a path that has the desired properties.
    
    Let $w$ be the node in $P_T$ that has smallest depth in $T_M$, i.e., the lowest common ancestor of $u$ and $v$ in $T_M$.
    We follow the \emph{right} subpath $P_r = (w = x_i, \ldots, x_m = v)$ of $P_T$ from $w$ to $v$ (the \emph{left} subpath $P_\ell$ from $u$ to $w$ is analogous).
    Starting at $w$, we sum the weights of the edges of $T_M$ on $P_r$ until we reach a node $x_j$ such that the sum is higher than $d_M(w, x_j)$.
    In this case, the edge $\{x_{j-2}, x_{j-1}\}$ must have skipped the node $x_j$, i.e., $x_j$ lies on the unique path from $x_{j-2}$ to $x_{j-1}$ in $M$.
    We now follow $P_r$ as long as we only move in the direction of $x_{j-2}$ in $M$, i.e., we move to the next node if that node is closer to $x_{j-2}$ in $M$ than the previous one, until we stop at a node $x_k$.
    By the definition of our algorithm, there must be an edge $\{x_{j-2}, x_k\} \in E_D$ with $w(\{x_{j-2}, x_k\}) = d_M(\{x_{j-2}, x_k\})$.
    We replace the subpath of $P_r$ from $x_{j-2}$ to $x_k$ by this edge, after which the length of the subpath of $P_r$ from $w$ to $x_k$ equals $d_M(x_i, x_k)$.
    We continue the process starting at $x_k$ until we reach $x_m$, and obtain that $w(P_r) = d_M(x_i, x_k)$.
    After we have performed the same process at $P_\ell$ (in the other direction), we have that $w(P_T) = d_M(u,v)$.
    Finally, note that $|P_T| = O(\log n)$, so $P_T$ has all the desired properties of the path $P$ from the statement of the lemma.
\end{proof}
}

For any node $v \in V$, we define the \emph{nearest shortcut node} of $v$ as $\sigma(v) = \argmin_{u \in \Sigma}d(v,u)$.
To let each node $v$ determine $\sigma(v)$ and $d(v, \sigma(v))$, we perform a distributed version of the Bellman-Ford algorithm.
From an abstract level, the algorithm works as follows.
In the first round, every shortcut node sends a message associated with its own identifier and distance value $0$ to itself.
In every subsequent round, every node $v \in V$ chooses the message with smallest distance value $d$ received so far (breaking ties by choosing the one associated with the node with highest identifier), and sends a message containing $d + w(\{v,u\})$ to each neighbor $u$ in $D_T$.
After $O(\log n)$ rounds, every node $v$ knows the distance $d_M(v,u)$ to its closest shortcut node $u$ in $M$.
Since for any closest shortcut node $w$ in $G$ there must be a shortest path from $v$ to $w$ that only contains edges of $M$, this implies that $u$ must also be closest to $v$ in $G$, i.e., $u = \sigma(v)$, and $d_M(v,u) = d(v,\sigma(v))$.

Note that each node has only created additional edges to its descendants in $T_M$ during the construction of $D_T$, therefore the degree of $D_T$ is $O(\log n)$ and we can easily perform the algorithm described above using the global network.

\begin{lemma} \label{lem:spa:nearest_shortcut}
    After $O(\log n)$ rounds, each node $v \in V$ knows $\id(u)$ of its nearest shortcut node $\sigma(v)$ in $G$ and its distance $d(v,\sigma(v))$ to it.
\end{lemma}

\lng{
\begin{proof}
    We first show that $v$ learns $\id(u)$ and $d_M(u,v)$ of its closest shortcut node $u$ in $M$ (if there are multiple, let $u$ be the one with highest identifier).
    Let $P_M = (u=x_0,x_1,\ldots,x_k=v)$ be the shortest path from $u$ to $v$ in $M$.
    Due to Lemma~\ref{lem:spa:distance_graph_correctness} (1), we know that $v$ will never receive a message with a smaller distance value than $d_M(u,v)$.
    Also, due to Lemma~\ref{lem:spa:distance_graph_correctness} (2), there is a path $P$ of length $O(\log n)$ from $u$ to $v$ with $w(P) = d_M(u,v)$.
    We claim that a message from $u$ traverses the whole path $P$ until it arrives at $v$ after $O(\log n)$ rounds.
    Assume to the contrary that this is not the case.
    Then there must exist a node $x_i$ on the path $P$ that does not send a message with $\id(u)$ and distance value $d_M(u,x_{i+1})$ to $x_{i+1}$.
    This can only happen if $x_i$ knows a shortcut node $z$ with $d_M(z,x_i) < d_M(u, x_i)$ such that $d_M(z,x_i) < d_M(v,x_i)$.
    This implies that $z$ is a shortcut node with $d_M(u,z) < d_M(u,v)$, contradicting the fact that $v$ is $u$'s nearest shortcut node.
    
    Finally, we show that $u$ is also the closest shortcut node $\sigma(v)$ to $v$ in $G$ (not only in $M$).
    Assume to the contrary that there is a shortcut node $w \in V$ such that $d(v, w) < d(v,u)$.
    If there is a shortest path from $v$ to $w$ in $G$ that does contain different shortcut nodes, then $w$ cannot be closest in $G$.
    Otherwise, the path contains only edges of $M$, which implies that $w$ is also closer to $v$ in $M$ than $u$, which contradicts our choice of $u$.
\end{proof}

Finally, the following lemma, which we will use later, implies that for each node $v$ there is at most one additional edge in $D_T$ to an ancestor in $T_M$.

\begin{lemma}\label{lem:spa:one_parent}
    If our algorithm creates an additional edge from $s$ to $t$, where $s$ is an ancestor of $t$ in $T_M$, then no node on the path from $s$ to $t$ in $T_M$ creates an additional edge to $t$.
\end{lemma}
}

\lng{
\begin{proof}
    Assume there exists an edge $\{s,t\} \in E_D \setminus E_T$ and let $P = (s,v_1,\ldots,v_k=t)$ be the unique path from $s$ to $t$ in $T_M$.
    Since $s$ established an edge to $t$, the path from $s$ to $v_1$ in $M$ must contain node $t$; furthermore, it contains all nodes $v_1, \ldots, v_k$.
    However, for each node $v_i$, $1 \le i \le k-1$, the path from $v_i$ to $v_{i+1}$ does not contain any other node of $P$.
    Therefore, our algorithm will not establish any additional edge in $E_D \setminus E_T$ from $v_i$ to any other node of $P$, including $t$.
\end{proof}
}

\subsection{Computing APSP between Shortcut Nodes} \label{sec:spa:sparse:apsp_cut_nodes}

In this section, we first describe how the shortcut nodes can compute their pairwise distances in $M$ by using $D_T$.
Then, we explain how the information can be used to compute all pairwise distances between shortcut nodes in $G$ by performing matrix multiplications.

\textbf{Compute Distances in $M$.}
First, each node learns the total number of shortcut nodes $n_c := |\Sigma|$, and each shortcut node is assigned a unique identifier from $[n_c]$.\shrt{\footnote{We denote $[k] = \{0, \ldots, k-1\}$.}}
The first part can easily achieved using Lemma~\ref{lem:spa:aggregation}.
For the second part, consider the Patricia trie $P$ on the node's identifiers, which, since each node knows all identifiers, is implicitly given to the nodes.
By performing a convergecast in $P$ (where each inner node is simulated by the leaf node in its subtree that has highest identifier), every inner node of $P$ can learn the number of shortcut nodes in its subtree in $P$.
This allows the root of $P$ to assign intervals of labels to its children in $P$, which further divide the interval according to the number of shortcut nodes in their children's subtrees, until every shortcut node is assigned a unique identifier.

Note that it is impossible for a shortcut node to explicitly learn all the distances to all other shortcut nodes in polylogarithmic time, since it may have to learn $\Omega(n^{1/3})$ many bits.
However, if we could distribute the distances of all $O(n^{2/3})$ pairs of shortcut nodes uniformly among all nodes of $V$, each node would only have to store $O(\log n)$ bits\footnote{In fact, for this we could even allow $n$ pairs, i.e., $n_c = O(\sqrt{n})$; the reason for our bound on $n_c$ will become clear later.}.
We make use of this in the following way.
To each pair $(i,j)$ of shortcut nodes we assign a \emph{representative} $h(i,j) \in V$, which is chosen using (pseudo-)random hash function $h: [n_c]^2 \rightarrow V$ that is known to all nodes and that satisfies $h(i,j) = h(j,i)$.\footnote{Note that sufficient shared randomness can be achieved in our model by broadcasting $\Theta(\log^2 n)$ random bits in time $O(\log n)$ \cite{AGG+19}.
Further, note that for a node $v \in V$ there can be up to $O(\log n)$ keys $(i,j)$ for which $h(i,j) = v$, w.h.p., thus $v$ has to act on behalf of at most $O(\log n)$ nodes.}
The goal of $h(i,j)$ is to infer $d_M(i,j)$ from learning all the edges on the path from $i$ to $j$ in $D_T$.

\lng{. 
To do so, the representative $h(i,j)$ first has to retrieve the labels of both $i$ and $j$ in $T_M$.
However, $i$ cannot send this information directly, as it would have to reach the representatives of every shortcut node pair $(i,k)$, of which there may be up to $\Omega(n^{1/3})$ many.
Instead, it performs a \emph{multicast} using techniques of \cite{AGG+19} to inform all these representatives.
To that end, $h(i,j)$ first joins the \emph{multicast groups} of $g(i)$ and $g(j)$.
Technically, it participates in the construction of \emph{multicast trees} in a simulated $\lfloor \log n\rfloor$-dimensional \emph{butterfly network} towards them.
More precisely, the nodes of the $i$-th column of the butterfly are simulated by the node with $i$-th highest identifier, and each source is a node of the butterfly's bottom level chosen uniformly and independently at random using a (pseudo)-random hash function $g$.
Since the construction is very technical, we leave out the details and defer the interested reader to \cite{AGG+19}.
By applying \cite[Theorem 2.3]{AGG+19} as a black box with parameters $\ell := O(\log n)$ and $L := O(n^{2/3})$ (since each node acts for at most $O(\log n)$ of the $O(n^{2/3})$ representatives, w.h.p., and each representative joins two multicast groups), we obtain multicast trees with \emph{congestion} $C = O(\log n)$ in time $O(\log n)$, w.h.p.
We then use \cite[Theorem 2.4]{AGG+19} to let each shortcut node $i$ multicast its label $l(i)$ to all representatives $h(i,k)$.
With parameter $\hat{\ell}$ as the maximum number of representatives simulated by the same node, multiplied by $2$ (which can easily be computed using Lemma~\ref{lem:spa:aggregation} on $M$), and congestion $C$, the theorem gives a runtime of $O(\log n)$, w.h.p.

From the knowledge of $l(i)$ and $l(j)$, $h(i,j)$ can easily infer the labels of all nodes on the path $P$ from $i$ to $j$ in $T_M$.
Specifically, it knows the label $l(x)$ of the highest ancestor $x$ of $i$ and $j$ in $T_M$, which is simply the longest common prefix of $l(i)$ and $l(j)$.\footnote{Technically, $h(i,j)$ could only infer the exact labels of the nodes of $P$ if it knew the degree of every node in $T_M$.
To circumvent this, $h(i,j)$ simply assumes that the tree is binary, which implies that some nodes of $P$ (apart from $i$, $j$, and $x$) may not actually exist.
However, as this is not a problem for the algorithm, we disregard this issue in the remainder of this section.}
The goal of $h(i,j)$ is to retrieve the edge from each node $v \in P \setminus \{x\}$ to its parent in $T_M$, as well as $v$'s additional edge to an ancestor in $D_T$, of which there can be at most one by Lemma~\ref{lem:spa:one_parent}.
Since by Lemma~\ref{lem:spa:distance_graph_correctness} these edges contain a shortest path from $i$ to $j$ that preserves the distance in $M$, $h(i,j)$ can easily compute $d_M(i,j)$ using this information.

To retrieve the edges, $h(i,j)$ joins the multicast groups of all nodes of $P \setminus \{x\}$ using \cite[Theorem 2.3]{AGG+19}.
Then, each inner node of $T_M$ performs a multicast using \cite[Theorem 2.4]{AGG+19} to inform all nodes in its multicast group about its at most two edges.
Since each node acts on behalf of at most $O(\log n)$ representatives, and each representative joins $O(\log n)$ multicast groups, the parameters of \cite[Theorem 2.3]{AGG+19} are $\ell := O(\log n)$ and $L := O(n^{2/3} \log n)$, and for \cite[Theorem 2.4]{AGG+19} we have $\hat{\ell} := O(\log n)$; therefore, all can be done in $O(\log n)$ rounds, w.h.p.
We conclude the following lemma.}
\shrt{ 
Due to space reasons, we defer a precise description to the full version.
From a high level, each $h(i,j)$ first needs to retrieve the labels of both $i$ and $j$ in $T_M$, which it cannot do directly, as the nodes may be contacted by many other nodes.
Instead, we use techniques from \cite{AGG+19} to distribute the load: $h(i,j)$  participates in the construction of a \emph{multicast tree} towards both $i$ and $j$.
Using randomization, these trees can be used to disseminate information from each node to all nodes in its multicast tree in a broadcast fashion with low congestion rather than communicating directly.
Afterwards, $h(i,j)$ can infer the labels of all nodes on the path from $i$ to $j$ in $D_M$, and learn their edge weights in a very similar way.
By observing that each node only has to learn $O(\log n)$ values, we can use Theorems 2.3 and Theorem 2.4 of \cite{AGG+19} in a straight-forward manner to obtain the following lemma. 
}

\begin{lemma} \label{lem:spa:representative_M-distance}
    Every representative $h(i,j)$ learns $d_M(i,j)$ in time $O(\log n)$, w.h.p.
\end{lemma}

\textbf{Compute Distances in $G$.}
Let $A \in \mathbb{N}_0^{n_c \times n_c}$ be the \emph{distance matrix} of the shortcut nodes, where \shrt{$A_{i,j} = \min\{w(\{i,j\}), d_M(i,j)\}$, if $\{i,j\} \in E$, and $A_{i,j} = d_M(i,j)$, otherwise.}\lng{\[
 A_{i,j} = \begin{dcases*}
        \min\{w(\{i,j\}, d_M(i,j)\} & if $\{i,j\} \in E$\\
        d_M(i,j) & otherwise.
        \end{dcases*}
\]} 
Our goal is to square $A$ for $\lceil \log n \rceil + 2$ many iterations in the \emph{min-plus semiring}.
More precisely, we define $A^1 = A$, and for $t \ge 1$ we have that $A^{2^t}_{i,j} = \min_{k \in [n_c]}(A^{2^{t-1}}_{i,k} + A^{2^{t-1}}_{k,j})$.
The following lemma shows that after squaring the matrix $\lceil \log n \rceil + 2$ times, its entries give the distances in $G$.

\begin{lemma}\label{lem:spa:apsp}
    $A^{2^{\lceil \log n \rceil + 2}}_{i,j} = d(i,j)$ for each $i,j \in \Sigma$.
\end{lemma}

\lng{
\begin{proof}
    First, note that $A_{i,j}^t \ge d(i,j)$ for all $i,j \in \Sigma$ and all $t \ge 1$, since every entry corresponds to the length of an actual path in $G$.
    We show that for all $t \ge 2$, $A^{2^t}_{i,j} \le \min_{P \in \mathcal{P}(i,j,t-2)} w(P)$, where $\mathcal{P}(i,j,t)$ is the set of all paths from $i$ to $j$ in $G$ that contain at most $2^{t}$ non-tree edges.
    Since any shortest path between two shortcut nodes $i$ and $j$ contains at most $n-1$ non-tree edges, $A^{2^{\lceil \log n \rceil + 2}}_{i,j} = d(i,j)$.
    
    To establish the induction base, we first show that for $t=2$, $A^{4}_{i,j} \le \min_{P \in \mathcal{P}(i,j,0)} w(P)$.
    Let $P$ be a path from $i$ to $j$ that contains at most $2^0 = 1$ non-tree edge $\{v_1, v_2\}$ such that $w(P) = \min_{P \in \mathcal{P}(i,j,0)}$.
    W.l.o.g., assume that $v_1$ appears first in $P$ ($v_1$ might be $i$).
    Let $P_1$ be the subpath of $P$ from $i$ to $v_1$, and $P_2$ be the subpath from $v_2$ to $j$.
    Note that $A_{i,v_1}^1 \le w(P_1)$, $A_{v_1, v_2}^1 \le w(e)$, and $A_{v_2, j}^1 \le w(P_2)$.
    Therefore, $A_{i,v_2}^2 \le w(P_1) + w(e)$ and $A_{v_2, j}^2 \le w(P_2)$, which implies that $A_{i,v_2}^4 \le w(P_1) + w(e) + w(P_2) = w(P) = \min_{P \in \mathcal{P}(i,j,0)}$.
    
    Now let $t > 2$ and consider a path $P \in \mathcal{P}(i,j,t-2)$ such that $w(P) = \min_{P \in \mathcal{P}(i,j,t-2)}$.
    Since $P$ contains at most $2^{t-2} \ge 2$ non-tree edges, we can divide $P$ at some shortcut node $k \in \Sigma$ into two paths $P_1$ and $P_2$ that both contain at most $2^{t-3}$ non-tree edges.
    We have that
    \begin{align*}
        A_{i,j}^{2^t} &\le A_{i,k}^{2^{t-1}} +  A_{i,k}^{2^{t-1}} \\
        &\le \min_{P \in \mathcal{P}(i,k,t-3)} w(P) + \min_{P \in \mathcal{P}(k,j,t-3)} w(P) \\
        &\le w(P_1) + w(P_2) \\
        &= w(P) = \min_{P \in \mathcal{P}(i,j,t-2)} w(P),
    \end{align*}
    which concludes the proof.

\end{proof}
}

We now describe how the matrix can efficiently be multiplied.
As an invariant to our algorithm, we show that at the beginning of the $t$-th multiplication, every representative $h(i,j)$ stores $A^{2^{t-1}}_{i,j}$.
Thus, for the induction basis we first need to ensure that every representative $h(i,j)$ learns $A_{i,j}$.
By Lemma~\ref{lem:spa:representative_M-distance}, $h(i,j)$ already knows $d_M(i,j)$, thus it only needs to retrieve $w(\{i,j\})$, if that edge exists.
To do so, we first compute an orientation with outdegree $O(\log n)$ in time $O(\log n)$ using \cite[Corollary 3.12]{BE10} in the local network.
For every edge $\{i,j\}$ that is directed from $i$ to $j$, $i$ sends a message containing $w(\{i,j\})$ to $h(i,j)$; since the arboricity of $G$ is $O(\log n)$, every node only has to send at most $O(\log n)$ messages.

The $t$-th multiplication is then done in the following way.
We use a (pseudo-)random hash function $h: [n_c]^3 \rightarrow V$, where $h(i,j,k) = h(j,i,k)$.
First, every node $h(i,j,k) \in V$ needs to learn $A^{2^{t-1}}_{i,j}$.\footnote{We will again ignore the fact that a node may have to act on behalf of at most $O(\log n)$ nodes $h(i,j,k)$.}
To do so, $h(i,j,k)$ joins the multicast group of $h(i,j)$ using \cite[Theorem 2.3]{AGG+19}.
With the help of \cite[Theorem 2.4]{AGG+19}, $h(i,j)$ can then multicast $A^{t-1}_{i,j}$ to all $h(i,j,k)$.
Since there are $L \le [n_c]^3 = O(n)$ nodes $h(i,j,k)$ that each join a multicast group, and each node needs to send and receive at most $\ell = O(\log n)$ values, w.h.p., the theorems imply a runtime of $O(\log n)$, w.h.p.

After $h(i,j,k)$ has received $A^{2^{t-1}}_{i,j}$, it sends it to both $h(i,k,j)$ and $h(j,k,i)$.
It is easy to see that thereby $h(i,j,k)$ will receive $A^{2^{t-1}}_{i,k}$ from $h(i,k,j)$ and $A^{2^{t-1}}_{k,j}$ from $h(k,j,i)$.
Afterwards, $h(i,j,k)$ sends the value $A^{2^{t-1}}_{i,k} + A^{2^{t-1}}_{k,j}$ to $h(i,j)$ by participating in an aggregation using \cite[Theorem 2.2]{AGG+19} and the minimum function, whereby $h(i,j)$ receives $A^{2t}_{i,j}$.
By the same arguments as before, $L = O(n)$, and $\ell = O(\log n)$, which implies a runtime of $O(\log n)$, w.h.p.
\lng{We conclude the following lemma.}

\begin{lemma} \label{lem:spa:matrix_mult}
    After $\lceil \log n \rceil + 2$ many matrix multiplications, $h(i,j)$ stores $d(i,j)$ for every $i, j \in [n_c]$.
    The total number of rounds is $O(\log^2 n)$, w.h.p.
\end{lemma}

\subsection{Approximating SSSP and the Diameter} \label{sec:spa:sparse:approx}
We are now all set in order to compute approximate distances between any two nodes $s,t \in V$.
Specifically, we approximate $d(s,t)$ by
\[
    \widetilde{d}(s,t) = \min\{d_{M}(s,t), d(s,\sigma(s)) + d(\sigma(s),\sigma(t)) + d(\sigma(t),t)\}.
\]
We now show that $\widetilde{d}(s,t)$ gives a $3$-approximation for $d(s,t)$.

\begin{lemma}\label{lem:spa:sparse:sssp_approx}
    Let $s,t \in V$ and $d(s,t)$ be the length of the shortest path from $s$ to $t$.
    It holds that $d(s,t) \leq \widetilde{d}(s,t) \leq 3d(s,t)$.
\end{lemma}

\lng{
\begin{proof}
    If the shortest path between $s$ and $t$ does not contain any shortcut node, then $\widetilde{d}(s,t) = d_{M}(s,t)=d(s,t)$.
    Assume the shortest path $P$ from $s$ to $t$ contains at least one non-tree edge $(x,y)$, such that $\widetilde{d}(s,t) = d(s,\sigma(s)) + d(\sigma(s),\sigma(t)) + d(\sigma(t),t)$.
    Then $P$ contains at least two shortcut nodes $x$ and $y$ with $d(s,t)=d(s,x)+d(x,y)+d(y,t)$.
    Consider Figure~\ref{fig:routing_example} for an illustration.
    \begin{figure}[t]
        \centering
        \includegraphics[]{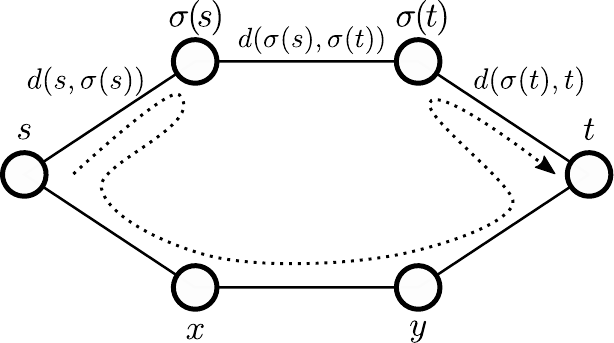}
    \caption{Illustration for the computation of the approximate distance between nodes $s$ and $t$.
    By the triangle inequality, the length $d(s,\sigma(s)) + d(\sigma(s), \sigma(t)) + d(\sigma(t), t)$ is at most the length of the dotted path, which, since $d(s,\sigma(s)) \leq d(s,x)$ and $d(\sigma(t),t) \leq d(y,t)$ has length at most $3d(s,t)$.}
    \label{fig:routing_example}
    \end{figure}

    Obviously, $\widetilde{d}(s,t)$ represents the distance of a path from $s$ to $t$, so $\widetilde{d}(s,t) \geq d(s,t)$ holds.
    Since $\sigma(s)$ is the nearest shortcut node of $s$ we get $d(s,\sigma(s)) \leq d(s,x)$ and, analogously, $d(\sigma(t),t) \leq d(y,t)$.
    Also, it holds
    \begin{align*}
        d(\sigma(s),\sigma(t)) &\leq d(\sigma(s),s) + d(s,x) + d(x,y) + d(y,t) + d(t,\sigma(t)) \\
        &= d(s,\sigma(s)) + d(s,t) + d(t,\sigma(t)).
    \end{align*}

    Putting all pieces together, we get
    \begin{align*}
        \widetilde{d}(s,t) &= d(s,\sigma(s)) + d(\sigma(s),\sigma(t)) + d(\sigma(t),t) \\
        &\leq d(s,\sigma(s)) + d(s,\sigma(s)) + d(s,t) + d(t,\sigma(t)) + d(\sigma(t),t) \\
        &= 2(d(s,\sigma(s)) + d(\sigma(t),t)) + d(s,t) \\
        &\leq 2(d(s,x) + d(y,t)) + d(s,t) \\
        &\leq 2d(s,t) + d(s,t) \\
        &= 3d(s,t).\qedhere
    \end{align*}
\end{proof}
}

To approximate SSSP, every node $v$ needs to learn $\widetilde{d}(s,v)$ for a given source $s$.
To do so, the nodes first have to compute $d_M(s,v)$, which can be done in time $O(\log n)$ by performing SSSP in $M$ using Theorem~\ref{thm:spa:sssp:tree}.
Then, the nodes construct $D_T$ in time $O(\log^2 n)$ using Lemma~\ref{lem:spa:distance_graph}.
With the help of $D_T$ and Lemma~\ref{lem:spa:nearest_shortcut}, $s$ can compute $d(s, \sigma(s))$, which is then broadcast to all nodes in time $O(\log n)$ using Lemma~\ref{lem:spa:aggregation}.
Then, we compute all pairwise distances in $G$ between all shortcut nodes in time $O(\log^2 n)$, w.h.p., using Lemma~\ref{lem:spa:matrix_mult}; specifically, every shortcut node $v$ learns $d(\sigma(s), v)$.
By performing a slight variant of the algorithm of Lemma~\ref{lem:spa:nearest_shortcut}, we can make sure that every node $t$ not only learns its closest shortcut node $\sigma(t)$ in $M$, but also retrieves $d(\sigma(s), \sigma(t))$ from $\sigma(t)$ within $O(\log n)$ rounds.
Since $t$ is now able to compute $\widetilde{d}(s,t)$, we conclude the following theorem.

\begin{theorem} \label{thm:spa:sssp_sparse}
    $3$-approximate SSSP can be computed in graphs that contain at most $n + O(n^{1/3})$ edges and have  arboricity $O(\log n)$ in time $O(\log^2 n)$, w.h.p.
\end{theorem}

For a $3$-approximation of the diameter, consider 
$
    \widetilde{D} = 2 \max_{s \in V} d(s, \sigma(s)) + \max_{x,y \in \Sigma}d(x,y).
$
$\widetilde{D}$ can easily be computed using Lemmas~\ref{lem:spa:distance_graph}, \ref{lem:spa:nearest_shortcut}, and \ref{lem:spa:matrix_mult}, and by using Lemma~\ref{lem:spa:aggregation} on $M$ to determine the maxima of the obtained values.
By the triangle inequality, we have that $D \le \widetilde{D}$.
Furthermore, since $d(s, \sigma(s)) \le D$ and $\max_{x,y \in \Sigma} d(x,y) \le D$, we have that $\widetilde{D} \le 3 D$.

\begin{theorem} \label{thm:spa:diameter_sparse}
    A $3$-approximation of the diameter can be computed in graphs that contain at most $n + O(n^{1/3})$ edges and have arboricity $O(\log n)$ in time $O(\log^2 n)$, w.h.p.
\end{theorem}

\bibliography{literature}

\lng{
\newpage 
\appendix

\section{PRAM Simulation} \label{sec:spa:pram_simulation}

Let $G$ be a graph with arboricity $a$ and let $\mathcal{A}$ be a PRAM algorithm that solves a graph problem on $G$ using $N$ processors with depth $T$.
Obviously, the total size of the input is $O(|E|)$.

\begin{lemma} \label{lem:spa:pram_sim}
    An EREW PRAM algorithm $\mathcal{A}$ can be simulated in the hybrid model in time $O(a/(\log n) + T \cdot (N/(n \log n) + 1) + \log n)$, w.h.p.
    A CRCW PRAM algorithm $\mathcal{A}$ can be simulated in time $O(a/(\log n) + T \cdot (N/n + \log n))$, w.h.p.
\end{lemma}

\begin{proof}
    Since in a PRAM the processes work over a set of shared memory cells $M$, we first need to map all of these cells uniformly onto the nodes.
    The total number of memory cells $|M|$ is arbitrary but polynomial and each memory cell is identified by a unique address $x$ and is mapped to a node $h(x)$, where $h: M \rightarrow V$ is a pseudo-random hash function.
    For this, we need shared randomness. 
    It suffices to have $\Theta(\log n)$-independence, for which only $\Theta(\log^2 n)$ bits suffice. 
    Broadcasting these $\Theta(\log^2 n)$ bits to all nodes takes time $O(\log n)$.
    
    To deliver $x$ to $h(x)$, the nodes compute an $O(a)$-orientation in time $O(\log n)$ \cite[Corollary 3.12]{BE10}.
    Note that each edge in $G$ can be represented by a constant amount of memory cells.
    When the edge $\{v, w\}$ that corresponds to $v$'s memory cell with address $x$ is directed towards $v$, $v$ fills in the part of the input that corresponds to $\{v, w\}$ by sending messages to all nodes that hold the corresponding memory cells (of which there can only be constantly many).
    Since each node has to send at most $O(a)$ messages, it can send them out in time $O(a / \log n)$ by sending them in batches of size $\lceil \log n \rceil$.
        
    We are now able to describe the simulation of $\mathcal A$: Let $k = n \lceil \log n \rceil$.
    Each step of $\mathcal{A}$ is divided into $\lceil N / k\rceil$ sub-steps, where in sub-step $t$ the processors $(t-1)k + 1, (t-1)k+2, \ldots, \min\{N, t k\}$ are active.
    Each node simulates $O(\log n)$ processors.
    Specifically, node $i$ simulates the processors $(t-1)k + (i-1)\lceil \log n \rceil + 1$ to $\min\{N, (t-1)k + i \lceil \log n \rceil\}$.
    When a processor attempts to access memory cell $x$ in some sub-step, the node that simulates it sends a message to the node $h(x)$, which returns the requested data in the next round.
    Since each node simulates $O(\log n)$ processors, each node only sends $O(\log n)$ requests in each sub-step.
    Also, in each sub-step at most $n \lceil \log n \rceil$ requests to distinct memory cells are sent in total as at most $n \lceil \log n \rceil$ are active in each sub-step.
    These requests are stored at positions chosen uniformly and independently at random, so each node only has to respond to $O(\log n)$ requests, w.h.p.
        
    In an EREW PRAM algorithm, the requests and responses can be sent immediately, since each memory location will only be accessed by at most one processor at a time.
    In this case, one round of the simulation takes time $O(N/(n \log n) + 1)$.
        
    In a CRCW PRAM algorithm, it may happen that the same cell is read or written by multiple processors.
    Thus, the processors cannot sent requests directly, but need to participate in aggregations towards the respective memory cells using techniques from \cite{AGG+19}.
    In case of a write, the aggregation determines which value is actually written; in case of a read, the aggregation is used to construct a multicast tree which is used to inform all nodes that are interested in the particular memory cell about its value.
    Since there can be only $O(n \log n)$ members of aggregation/multicast groups, and by the argument above each node only participates and is target of $O(\log n)$ aggregations (at most one for each processor it simulates), performing a sub-step takes time $O(\log n)$, w.h.p., by \cite{AGG+19}.
    Thus, each step can be performed in time $O(N/n + \log n)$, w.h.p. (note that the additional $\log n$-overhead stems from the fact in case $N > n$, one single node still needs time $O(\log n)$ to simulate a sub-step).
\end{proof}
}

\end{document}